\documentclass[11pt]{amsart}
\usepackage[T1]{fontenc}
\usepackage[legalpaper, margin=1in]{geometry}

\pdfoutput=1

\usepackage{amsmath}
\usepackage{graphicx}
\usepackage{subcaption}
\usepackage{hyperref}
\usepackage{float}

\usepackage{tikz}
\usepackage{tikz-3dplot}
\usetikzlibrary{arrows.meta,arrows}
\usetikzlibrary{shapes.geometric}

\theoremstyle{theorem}
\newtheorem{theorem}{Theorem}[section]

\newtheorem{lemma}[theorem]{Lemma}
\newtheorem{corollary}[theorem]{Corollary}

\theoremstyle{definition}

\newtheorem*{exa}{Example}

\begin{document}

\title[Evolving the Euler rotation axis as a dynamical system]{Evolving the Euler rotation axis as a dynamical system, using the Euler vector and generalizations}
\markright{ Euler Vector Dynamical System}

\author{John H. Elton}  
\address{School of Mathematics, Georgia Institute of Technology, Atlanta, GA}
\email{Jelton@bellsouth.net}

\author{John R. Elton}  
\email{Jelton.physics@gmail.com}



\begin{abstract}
Differential equations are derived which show how generalized Euler vector representations of the Euler rotation axis and angle for a rigid body evolve in time; the Euler vector is also known as a rotation vector or axis-angle vector. The solutions can exhibit interesting rotational features in this non-abstract, visualizable setting, including spinor-like behavior and quasiperiodicity.  The equations are well-behaved at zero, reducing to the simple infinitesimal case there. One of them is equivalent to a known quaternion differential equation. The simple geometric derivation does not depend on Euler's rotation theorem, and yields a proof of Euler's theorem using only infinitesimal motions.  With mild regularity conditions on the angular velocity function, there is a continuous evolution of the normalized axis and angle for all time. Dynamical systems properties are discussed, and numerical solutions are used to investigate them when the angular velocity is itself rotating, and the Euler vector trajectory traces out a torus-like shape, with a strobe plot that densely fills in a closed curve.
\end{abstract}

\maketitle

\section{\textbf{Introduction.}}
\label{sec:Introduction}

In 1775, Euler proved that any displacement of a rigid body in 3-dimensional space that leaves a point fixed, is a rotation by some angle about some axis through the fixed point. His proof was purely geometric.  Elementary physics books just implicitly assume it is true without bringing it up, or declare it is intuitively obvious (e.g. Marion \cite{marion} p. 362), which seems unlikely since Euler thought it required proof. Mechanics books now (e.g. Goldstein et. al \cite{goldstein}, Corben and Stehle \cite{corben}) show the existence of the rotation axis as the eigenvector with eigenvalue 1 of an orthogonal transformation that preserves orientation; some older ones (e.g. Whittaker \cite{whittaker}, Banach \cite{banach}) give geometric proofs, different from Euler’s\footnote{Palais et. al \cite{palais} provide an English exposition of Euler’s proof, along with their own novel proof that uses modern mathematical notation and rigor.  They also give a brief survey of several other proofs, and references.  Joseph \cite{joseph} gives yet another geometric proof, and references to other proofs, as well as a discussion of the extension of Euler’s theorem to general rigid body motion, to include translations.} \footnote{The theorem generally called Chasles’ theorem states that the general displacement of a rigid body can be split into a translation followed by rotation about an axis.  The Mozzi-Chasles theorem shows that the rotation axis can be chosen so that the translation is along the direction of that axis, called a screw axis.  It seems Giulio Mozzi showed the existence of the screw axis for a rotation-translation in 1763 \cite{ceccarelli},  even before Euler’s rotation theorem.  Chasles’ writing on this was later, in 1830, but most books attribute the result to Chasles and don’t mention Mozzi.  See \cite{ceccarelli} for a discussion.}.

Upon reflection, after making our own geometric proof of Euler’s rotation theorem, it occurred to us that our proof, as well as the one given by Euler and others, including the usual modern eigenvector proof, look at the original position of the body, and then a second position a finite (not infinitesimal) time later, without tracking how it moved there continuously. It would seem as though the rotation axis were decomposed and then magically reassembled in the final position, the only constraints for this operation being that the point at the origin goes back to the same point, the distance between any two points remains the same - and also that orientation is preserved.  The definition of rigid body motion given in texts asks that the distance between points does not change at any time during the motion, but does not explicitly say anything about preserving orientation.  Did they leave something out?  No, because motion means that it moved continuously from its starting to its final position, all the while preserving distance between points, and this forces orientation to be preserved at all times; see e.g. \cite{corben} p.137.

So it doesn’t seem right, philosophically, to simply look at the final position and say “I can prove that there is a rotation axis and angle!”. There was in fact a rotation axis and angle at each time from initial to final.  Thus our goal became \textit{finding the time evolution of the Euler rotation axis and angle}, starting from the conceptually simple infinitesimal case, and using only analysis of infinitesimal motions. In this way, we will find differential equations for the Euler vector (also known as the rotation vector or axis-angle vector) and generalizations, whose directions show the evolving axis of rotation, and whose lengths are functions of the angle of rotation. One of these equations is the real part of a known quaternion differential equation, obtained in a different way. The Euler vector itself has length equal to the rotation angle. With mild regularity conditions on the angular velocity function, these lead to continuous solutions for both the normalized axis vector and the rotation angle, for all time. Looking at them as dynamical systems, they can exhibit interesting behaviors, including spinor-like features and quasiperiodicity, despite arising from this simple non-abstract setting. 

As visual motivation for what follows, we first present Figures \ref{fig:Euler vec const w trajectory}, \ref{fig:Euler vec const w ind variables}, and \ref{fig:Euler vec rotating w}, in which we show examples of numerical solutions of the differential equation where the trajectory of the Euler rotation axis has surprisingly non-simple behavior. The explanation of these results will be given in the following sections, but it is worth having a picture in mind from the start. In Figures \ref{fig:Euler vec const w trajectory} and \ref{fig:Euler vec const w ind variables} below, a point on a rigid body is indicated by vector $\boldsymbol{B}$, and $\boldsymbol{B}$ rotates in the x-y plane at constant rate about a constant unit-length angular velocity vector $\pmb{\omega }$, which is along z; the initial Euler rotation axis is not parallel to $\pmb{\omega }$.  We observe that the trajectory of the Euler rotation axis, indicated by $\boldsymbol{\hat{E}}$, the normalization of the Euler vector solution $\boldsymbol{E}$, revolves in a plane that is tilted relative to the path of $\boldsymbol{B}$, with a period 4$\pi$ that is \textit{twice} the period with which the body revolves, and also that ${\boldsymbol{\hat{E}}}(t+2\pi) = -{\boldsymbol{\hat{E}}}(t)$.  Because of this spinor-like behavior (see e.g.\cite{mtw} p.1148), we call it the \textit{Euler Spinor}. An animation showing the continuous evolution of this example is available  \href{https://giphy.com/gifs/VPGMaYyGtF8TqxJy6S}{here}.
In Figure \ref{fig:Euler vec rotating w}, we observe the complex torus-like shape of the trajectory traced out by the Euler vector after a long time when $\pmb{\omega }(t)$ \textit{itself} rotates in a plane at constant frequency. A Poincar\'e section strobe plot (e.g. Strogatz \cite{strog} Example 12.5.2), sampled at the period with which $\pmb{\omega }$ rotates, is shown as well. As the integration time is increased, the strobe plot appears to be densely filling in a closed curve, a hallmark of quasiperiodicity (\cite{ivchenko}, \cite{broer}, \cite{das}). For the rotating $\pmb{\omega}$ case, we provide another animation of the motion of all 3 vectors, $\pmb{\omega}$, $\boldsymbol{B}$, and the rotation axis $\boldsymbol{\hat{E}}$, for the early part of the trajectory 
\href{https://giphy.com/gifs/5NInitbDR3EdOw1LYz}{here}.  Early on, the motion of  $\boldsymbol{\hat{E}}$ appears to be rather "random" and unpredictable, so it is fascinating that eventually, after a long time, the curve ends up tracing out the ordered-looking shape in Figure \ref{fig: E torus norm}! It is tempting to think that perhaps the trajectory is chaotic, though we find that this is not the case. By the end of the animation, one can start to see the final shape taking form. Another plot of the exotic trajectory of $\boldsymbol{E}$ for the case of $\pmb{\omega }$ rotating with a different (faster) frequency is shown in Figure \ref{fig:E and strobe period pi} and discussed later in the text.

For clarity, as solutions are visualized, it is important to keep in mind the various movements that are happening simultaneously: A rigid body rotates about a point that is fixed in the body frame with (generally non-constant) angular velocity $\pmb{\omega }(t)$; the position of a point on the rigid body is labeled by a vector $\boldsymbol{B}$; and finally, the Euler vector evolves in time as the body rotates, with ${\boldsymbol{\hat{E}}}(t)$ tracing out a trajectory which, at any given moment, represents the axis about which the body would have to be rotated from its initial position to get into the current position, with a \textit{single} rotation by angle $|{\boldsymbol{E}}|$ . 

The remainder of this paper is organized as follows:

Section \ref{sec:EulerVec} introduces the Euler vector as an extension of the usual representation of infinitesimal rotation that comes from the existence of angular velocity, suitable for our purposes of growing finite motions from infinitesimal ones.

Section \ref{sec:DiffEqns} gives a simple informal geometric argument with infinitesimals to arrive at the differential equations that are the subject of this paper. Then we give a rigorous proof that the solutions do give correct representations of the motion.  This yields a proof of Euler's theorem using only infinitesimal arguments, as we set out to do. The Modified Gibbs version of the Euler vector is arrived at as giving the simplest equation, and is shown to be equivalent to a quaternion differential equation.  This section also gives an example of a pathological angular velocity function that demonstrates that, without some regularity condition, it is not always possible to have the normalized axis direction vector evolve continuously.

Section \ref{sec: spinor} treats the case that the angular velocity is constant after some initial motion, as in Figures \ref{fig:Euler vec const w trajectory} and \ref{fig:Euler vec const w ind variables}.  A closed form for the solution in this case is given, with a geometric interpretation of the trajectory of ${\boldsymbol{\hat{E}}}$ being the normalization of a simply-parametrized planar ellipse, and with spinor-like properties.

Section \ref{sec: boundary} shows that with a mild regularity condition on the angular velocity function, the normalized rotation axis and rotation angle evolve continuously for all time, by allowing the rotation angle to go beyond $2\pi$ and below zero.  We give an example to show that this can happen in a non-trivial way, though it does not seem to occur naturally in our numerical experiments. The conditions also imply an Euler vector representation for all time, including when the rotation angle takes values that are multiples of $2\pi$. 

Finally, in Section \ref{sec:NumericalSolns}, some numerical solutions of the dynamical system are investigated for a case in which the angular velocity of the body is \textit{not} constant in direction, and in this more complex case it is seen that the Euler vector traces out quasiperiodic trajectories.

\begin{figure}
\centering
	\begin{subfigure}{0.49\textwidth}
		\includegraphics[width=1.0\textwidth,height=1.0\linewidth]{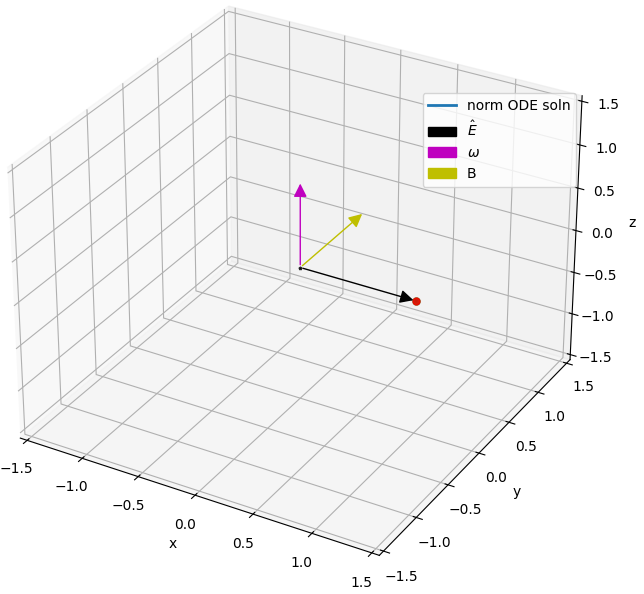}
		\caption{Initial conditions}
		\label{fig:E_0}
	\end{subfigure}
	\begin{subfigure}{0.49\textwidth}
		\includegraphics[width=1.0\textwidth,height=1.0\linewidth]{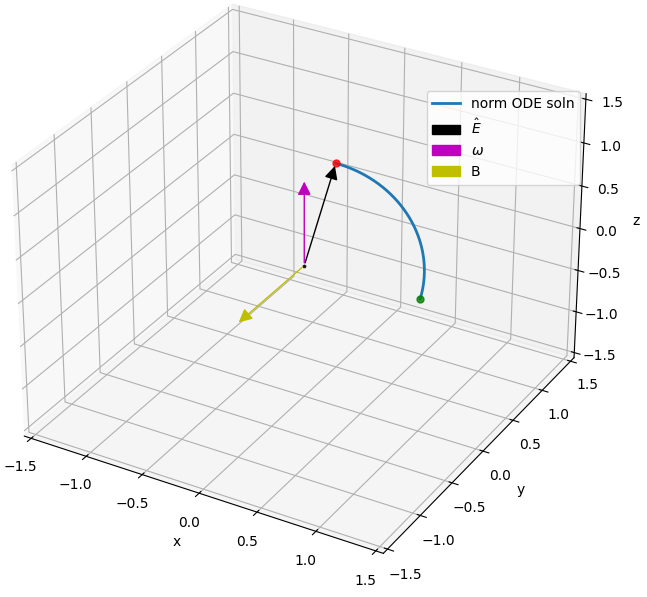}
		\caption{After time $\pi$}
		\label{fig:E_1}
	\end{subfigure}
	\begin{subfigure}{0.49\textwidth}
		\includegraphics[width=1.0\textwidth,height=1.0\linewidth]{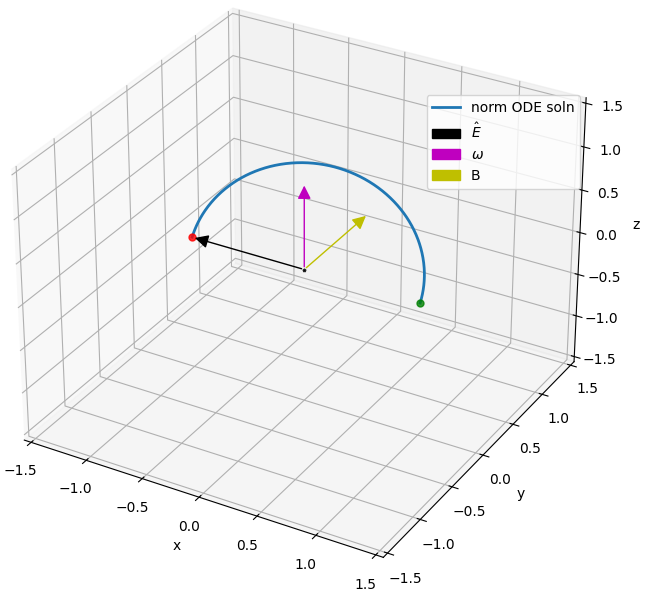}
		\caption{After time 2$\pi$}
		\label{fig:E_2}
	\end{subfigure}
	\begin{subfigure}{0.49\textwidth}
		\includegraphics[width=1.0\textwidth,height=1.0\linewidth]{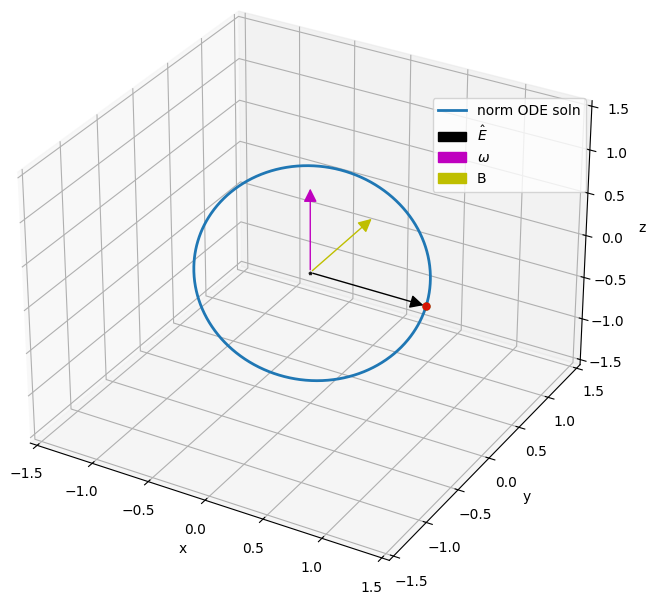}
		\caption{After time 4$\pi$}
		\label{fig:E_4}
	\end{subfigure}

	\caption{Evolution of the normalized Euler vector $\boldsymbol{\hat{E}}$ from times 0 to $\pi$, 2$\pi$, 4$\pi$ for constant $\pmb{\omega}$. Vectors indicating the location of a rigid body and the angular velocity vector are also plotted. The body vector $\boldsymbol{B}$ rotates at a constant rate in the x-y plane around $\pmb{\omega}$, which is along z. The path traced out by the Euler vector makes a complete periodic revolution in a plane that is tilted relative to the path of $\boldsymbol{B}$ in time 4$\pi$, while the body makes two complete revolutions, showing spinor-like behavior.}
	\label{fig:Euler vec const w trajectory}
\end{figure}

\begin{figure}
\centering
	\begin{subfigure}{0.49\textwidth}
		\includegraphics[width=1.0\textwidth,height=1.0\linewidth]{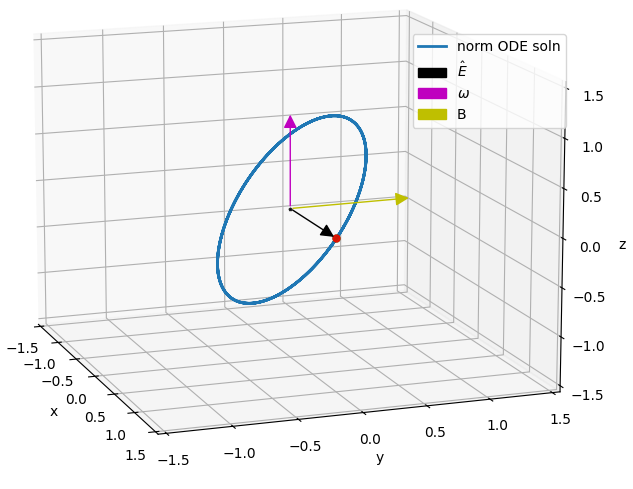}
		\caption{Trajectory of $\boldsymbol{\hat{E}}$}
		\label{fig:E 4 periods}
	\end{subfigure}
	\begin{subfigure}{0.49\textwidth}
		\includegraphics[width=1.0\textwidth,height=1.0\linewidth]{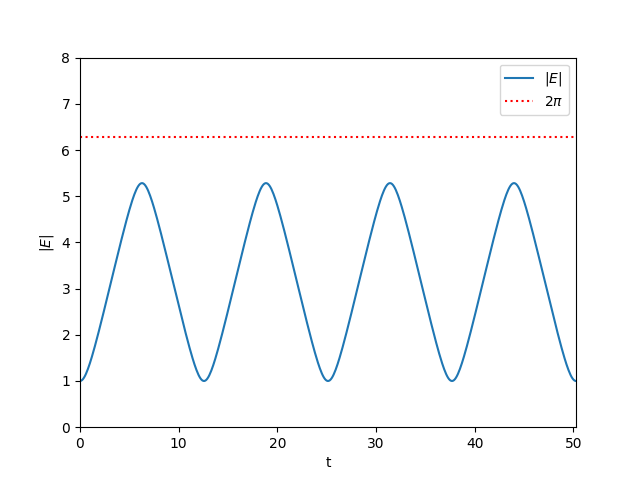}
		\caption{$|\boldsymbol{E}|$ against time}
		\label{fig: E t 4 periods}
	\end{subfigure}
	\begin{subfigure}{0.49\textwidth}
		\includegraphics[width=1.0\textwidth,height=1.0\linewidth]{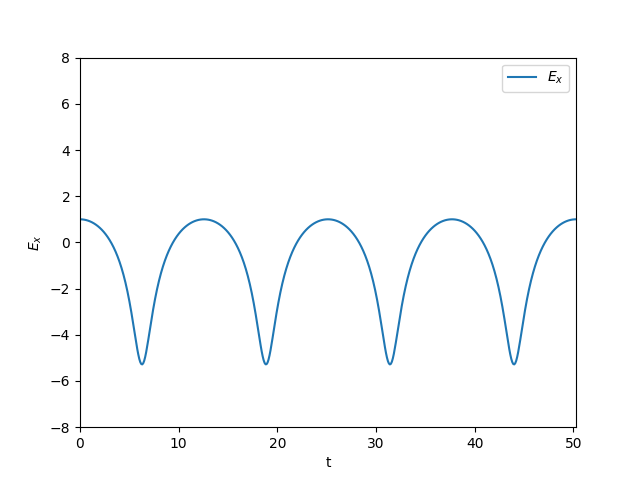}
		\caption{$E_x$ against time}
		\label{fig: Ex t 4 periods}
	\end{subfigure}
 	\begin{subfigure}{0.49\textwidth}
		\includegraphics[width=1.0\textwidth,height=1.0\linewidth]{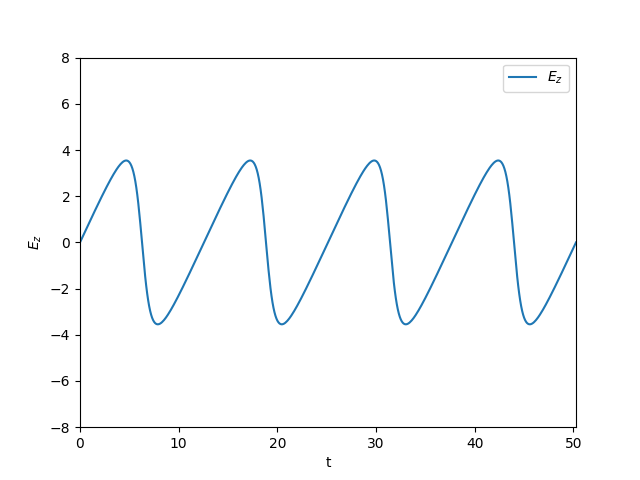}
		\caption{$E_z$ against time}
		\label{fig: Ez t 4 periods}
	\end{subfigure}

	\caption{Evolution of the Euler vector trajectory through several periods for constant $\pmb{\omega}$. We give another view showing the path of the tilted plane for the trajectory of $\boldsymbol{\hat{E}}$ while $\boldsymbol{B}$ rotates in the x-y plane. We also show a time series of the length of the unnormalized $\boldsymbol{E}$ and two components of $\boldsymbol{E}$, demonstrating the periodicity. }
	\label{fig:Euler vec const w ind variables}
\end{figure}

\begin{figure}
\centering
	\begin{subfigure}{0.49\textwidth}
		\includegraphics[width=1.0\textwidth,height=1.0\linewidth]{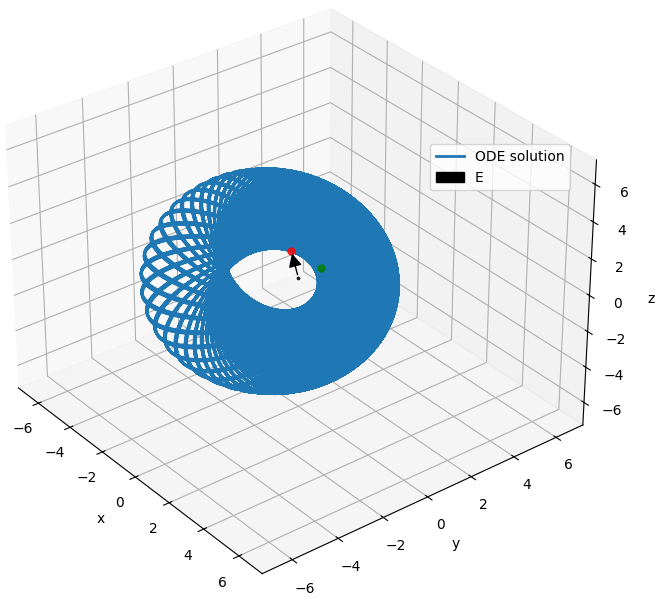}
		\caption{Trajectory of $\boldsymbol{E}$}
		\label{fig: E torus 40}
	\end{subfigure}
	\begin{subfigure}{0.49\textwidth}
		\includegraphics[width=1.0\textwidth,height=1.0\linewidth]{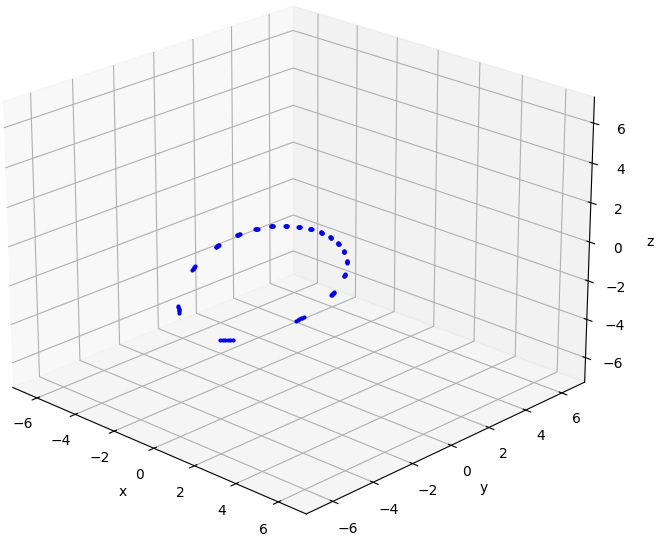}
		\caption{Strobe section plot for $\boldsymbol{E}$}
		\label{fig: E strobe 420 40}
	\end{subfigure}
	\begin{subfigure}{0.49\textwidth}
		\includegraphics[width=1.0\textwidth,height=1.0\linewidth]{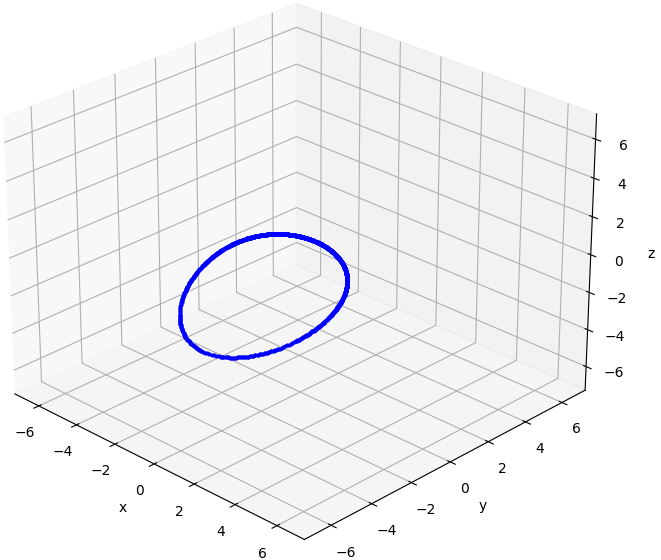}
		\caption{Strobe section plot for $\boldsymbol{E}$ at longer time}
		\label{fig: E strobe 4200 40}
	\end{subfigure}
	\begin{subfigure}{0.49\textwidth}
		\includegraphics[width=1.0\textwidth,height=1.0\linewidth]{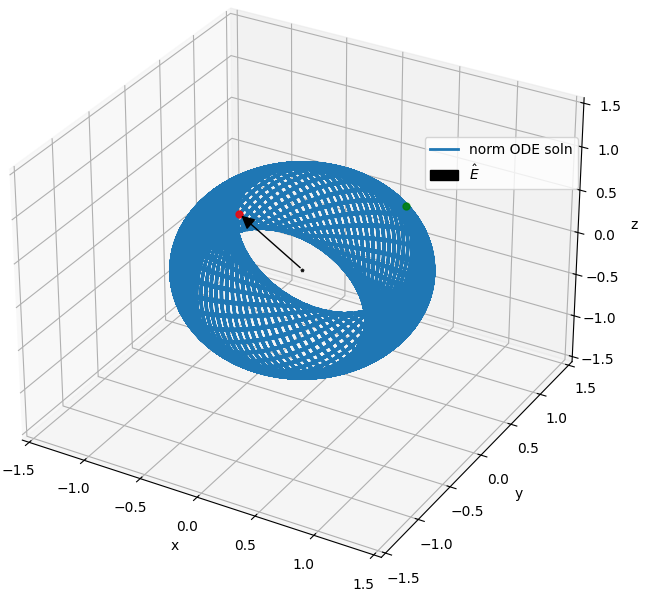}
		\caption{Trajectory of $\boldsymbol{\hat{E}}$}
		\label{fig: E torus norm}
    \end{subfigure}

	\caption{Evolution of the Euler vector trajectory for a rotating $\pmb{\omega}$($t$) with period $T=40$ up to time $t=4200$. In (a), we see the intricate torus-like path traced out by $\boldsymbol{E}$ over this long time. In (b), we view a Poincar\'e strobe plot of the trajectory in (a), sampled at the period with which $\pmb{\omega}$ rotates. (c) shows another strobe plot but for 10x as long of an integration time, showing the strobe section plot densely filling in a closed curve, an indication of quasiperiodicity. (d) shows the trajectory for the \textit{normalized} rotation axis, $\boldsymbol{\hat{E}}$.}
	\label{fig:Euler vec rotating w}
\end{figure}

\section{\textbf{From infinitesimal to finite. The Euler vector.}}
\label{sec:EulerVec}

A rigid body is a collection of points in $\mathbb{R}^3$ whose mutual distances remain constant in time; it is assumed that it has at least 3 non-collinear points.  The motion of the points is assumed continuous.  To speak more precisely about the motion of the body, we attach a cartesian coordinate system to the body, called the body system or frame, which moves along with the body in a space cartesian coordinate system which is not moving.  “Attach” means the points in the body have constant coordinates in the body system.  In the Euler rotation theorem, the fixed point need not be in the original physical body.  It is common, for example, for the center of mass to be  the fixed point, though that is not necessarily a point of the original physical rigid body.

From now on, we assume that the origin of the body system is the same as the origin of the space system at all times, and that at time $t=0$, the two systems coincide.  Quantities indicated with bold-faced vector notation will refer to the space coordinate system, unless stated otherwise.  The purpose of the Euler rotation theorem is to produce a rotation axis in the original space system.

We start with infinitesimal motions, which are easy to treat.  The following does not depend on Euler’s rotation theorem.

\begin{lemma}
[Existence of angular velocity]  For any rigid body motion that leaves the origin of the body system fixed, there exists a unique (axial) vector function $\pmb{\omega }(t)$ such that for any point in the body, if $\bf{r}(t)$  is its displacement at time $t$, then
\begin{equation}
\label{eqn: Ang velocity}
d\bold{r}/dt = \bold{\dot r}(t) = \pmb{\omega }(t) \times \bold{r}(t).
\end{equation}
We often suppress $t$ and write ${\bf{\dot r}} = \pmb{\omega } \times \bf{r}$.
\end{lemma}

The existence of angular velocity is either assumed or proved (with varying amounts of care) in every physics book (e.g. \cite{goldstein} p. 172).  A rigorous proof shows that $\pmb{\omega }(t)$ is continuous if the points in the body have continuous velocities.  We shall assume that $\pmb{\omega }(t)$ is continuous.

Let $d\pmb{\theta } = \pmb{\omega }dt$, so $d{\bf{r}} = d\pmb{\theta } \times {\bf{r}}$.  Thus, infinitesimal motion is indeed rotation, with rotation axis determined by the direction of $d\pmb{\theta }$, and angle $\left| {d{\pmb{\theta }}} \right|$ counterclockwise using the right-hand rule, with thumb pointing in direction of $\pmb{\omega}$.  Everyone is happy to call this a vector, more precisely an axial vector, representing the infinitesimal motion.

For \textit{ non-infinitesimal} rotations of a rigid body, one can do the same thing.   Suppose ${\boldsymbol{\hat{E}}}$ is a unit vector pointing along the axis of rotation, and $\theta $ is the counterclockwise angle of rotation about that axis. Define a finite analog of the infinitesimal rotation vector as
\begin{equation}
\label{Euler vector def 1}
\boldsymbol{E} = {\boldsymbol{\hat{E}}}\theta.
\end{equation}
This is called an \textit{Euler vector} in (\cite{wikiRotationformalisms}, \cite{wikiEulerVec}) and we will use that terminology.  It is also called a rotation vector (\cite{kruglov}, \cite{shuster},\cite{wikiRotationformalisms}), or axis-angle vector \cite{terzakis}. It is referred to as the \textit{Rodrigues parameter} in \cite{kruglov}, although that term is defined differently in \cite{terzakis}.  But one thing they all do is apologize for calling this a vector.   If $\boldsymbol{E}_1$ and $\boldsymbol{E}_2$ are Euler vector representations of two rotations, then $\boldsymbol{E}_1 + \boldsymbol{E}_2$ is generally not an Euler vector representation of the result of doing the two rotations in succession (in either order).  On account of this, \cite{kruglov} p. 236, says that the Euler vector is not a "real vector", and \cite{wikiRotationformalisms} says that this "shows that finite rotations are not really vectors at all".  The physics books don’t even write it down; \cite{goldstein}, p. 163, says that “such a correspondence cannot be made successfully”.  

We disagree.  Although it is true that finite rotations “are” not literally vectors (they are not literally matrices either: they are operations), that doesn’t mean they can’t be represented by vectors.  One property of the Euler vector that attests to its vector quality is that when drawn as an arrow in the real physical space of the rigid body, one sees that it is more than some arbitrary parameterization of the three degrees of freedom of the body: one can view the rotation being around the axis defined by that arrow in the picture, so its direction has real meaning, not an abstraction.  Another test of being a vector is how it transforms.  Under an orthogonal transformation, the rotation axis for a rigid body clearly just transforms along with the body.  If the orthogonal transformation is orientation-preserving, the counterclockwise sense of rotation stays the same, so the Euler vector is transformed the same way as the vectors in the body.  If the transformation is not orientation-preserving (it's determinant is -1), the transformed Euler vector has to be flipped, because the meaning of counterclockwise got reversed.  In other words, the Euler vector transforms as an axial vector, just as does the angular velocity.

Finally, an analogy: the unit vectors are a subset of $\mathbb{R}^3$ (the unit ball).  The sum of two of them is not a unit vector, but no one has ever said that unit vectors aren't "really" vectors on account of that.  They form a subset, not a subspace, but they are still vectors, elements of $\mathbb{R}^3$ . But if the reader prefers to insist that the Euler vector is not truly a vector, it will not affect the discussion that follows.  The reason we bring it up is that perhaps the dynamics of the Euler vector had not been previously considered due to its vector status being discounted.  But it is a natural quantity for passing from infinitesimal to finite motions because the vector that is always used to represent an infinitesimal motion is already in Euler vector form.  In the next section, we also consider generalizations of the Euler vector which behave just like the Euler vector in the infinitesimal limit.

\section{\textbf{Differential equations for the evolution of Euler vector and other vector representations of the motion.}}
\label{sec:DiffEqns}
$ $\\
Let $Rot_{{\bf{n}},\theta}$ be the operation of counterclockwise rotation by angle $\theta $ around the axis through the origin with direction unit vector $\bf{n}$; if $\theta$ is a multiple of $2\pi$, it is understood to be the identity operator even if $\bf{n}$ is undefined. The rotation can be represented by Euler vector $\boldsymbol{E} = \bf{n}\theta$. In (\cite{wikiRotationformalisms}, \cite{wikiEulerVec}), $\theta$ is called the length of the vector, which would make it non-negative. But $Rot_{{\bf{n}},\theta}$ makes sense when $\theta$ is negative; a counterclockwise rotation by a negative amount is just a clockwise rotation by that amount. Since ${\bf{n}}\theta = ({-\bf{n}})(-\theta)$,  $\boldsymbol{E} = \bf{n}\theta$ could always be written as a unit vector times the length of $\bf{E}$, with the length representing the rotation angle, if desired. But later we'll have occasion to need $\theta$ to be negative, in describing a continuous evolution of the Euler rotation axis and angle.

More generally, let $\boldsymbol{F} = {\bf{n}}f(\theta)$, where $f$ is odd with continuous third derivative, and $f'(0)>0$, so $f(\theta) = f'(0)\theta + f^{(3)}(0)\theta^3/6 + o(\theta^3)$.  We call this a \textit{generalized Euler vector representation}. It behaves similarly to the Euler vector for $\theta$ near zero. 
 We'll only be considering two examples in addition to the Euler vector case. 

Recall the following well-known formula.
\begin{lemma}
[Rodrigues' rotation formula] For any point $\bf{r}$, 
\begin{equation}
\label{eqn: Rod formula 1}
Rot_{{\bf{n}},\theta}(\bf{r}) = \bf{r}\cos \theta  + \bf{n}(\bf{n} \cdot \bf{r})(1 - \cos \theta ) + (\bf{n} \times \bf{r})\sin \theta .
\end{equation}
\end{lemma}
See e.g.(\cite{goldstein}, p. 162).  It is an easy trig-geometry exercise.

We now give a geometric argument with differentials to get a differential equation for a generalized Euler vector representation of the motion of the body.  We’ll later give a rigorous proof that it is correct, but the picture argument gives insight into why it  evolves the way it does.  Suppose at time $t$, the body position corresponds to rotation by angle $\theta $ about the axis with unit direction vector $\bf{n}$, so ${\bf{F}} = {\bf{n}}f(\theta)$. Note that this amounts to an induction assumption; we are not assuming Euler's rotation theorem.  We'll show that after an infinitesimal motion, the position still corresponds to a rotation. For a point initially at $\bf{r}$, ${\bf{r}}'= Rot_{{\bf{n}},\theta}(\bf{r})$ is its position at time $t.$   By Equation~\ref{eqn: Ang velocity}, its position at time $t + dt$ is ${\bf{r}}'' = {\bf{r}}' + {\pmb{\omega }}dt \times {\bf{r}}'$, the result of infinitesimal motion starting at ${\bf{r}}'$, where ${\pmb{\omega }} = {\pmb{\omega }}(t)$ is the angular velocity (we ignore $o(dt)$ terms, and assume ${\bf{r}}$ is from some bounded set).   The goal is to find ${\bf{\tilde F}} = {\bf{F}} + d{\bf{F}}$ so that for all ${\bf{r}},$ ${\bf{\tilde r}}= Rot_{{\bf{\tilde n}},\tilde \theta}(\bf{r})$ will coincide (up to $o(dt)$) with ${\bf{r}}'',$ where $\tilde {\bf{F}}=\tilde {\bf{n}}{f(\tilde \theta)}.$ That is, the result of an infinitesimal motion after a rotation is again a rotation, with representing vector ${\bf{F}} + d{\bf{F}}.$ To show how to find $d{\bf{F}}$, we separate into two cases: ${\pmb{\omega }}$ parallel to ${\bf{F}},$ and ${\pmb{\omega }}$ perpendicular to ${\bf{F}},$ and then put them together.
  
If ${\pmb{\omega }}$ is parallel to ${\bf{F}}$, then clearly $d{\bf{F}} = {\pmb{\omega }}f'(\theta)dt$: the rotation axis remains parallel to ${\pmb{\omega }}$; just the amount of rotation, and thus the length of ${\bf{F}}$, changes, in time $dt$.

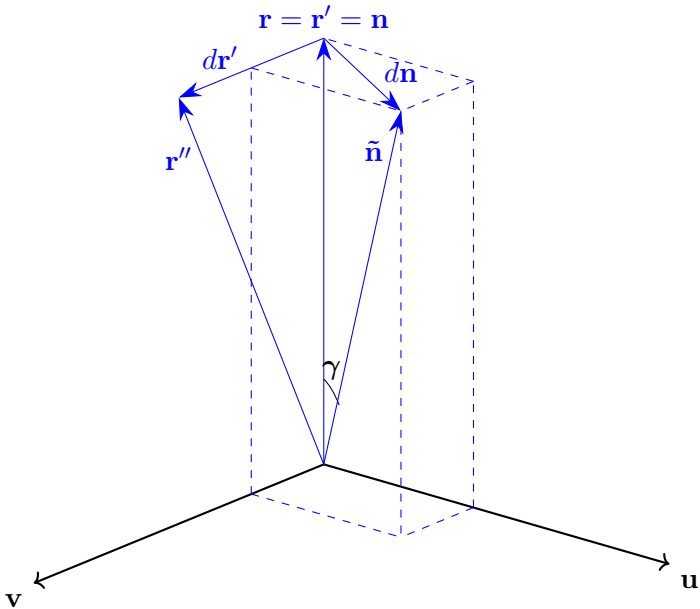
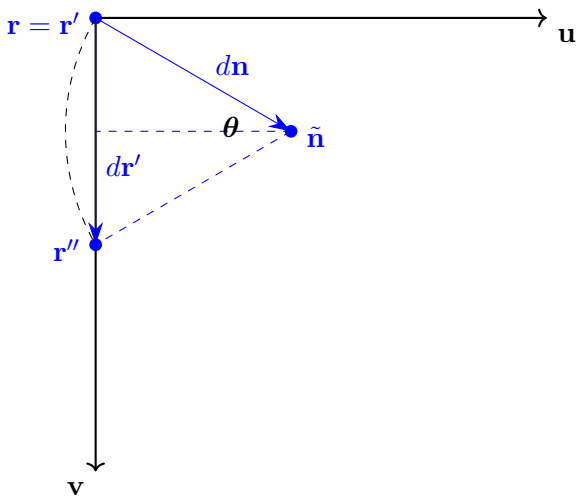
\begin{figure}[h!]
\centering
	\begin{subfigure}{0.85\textwidth}
		\tdplotsetmaincoords{70}{130}
        \begin{tikzpicture}[scale=6, tdplot_main_coords]
        \draw[thick,->] (0,0,0) -- (1,0,0) node[anchor=north east]{$\bf{v}$};
        \draw[thick,->] (0,0,0) -- (0,1,0) node[anchor=north west]{$\bf{u}$};
        \draw[-{Stealth[length=3mm, width=2mm]},color=blue] (0,0,0) -- (0,0,1) node[anchor=south]{$\bf{r}=\bf{r'}=\bf{n}$};

        \pgfmathsetmacro{\rvec}{1.0};
        \pgfmathsetmacro{\thetavec}{30};
        \pgfmathsetmacro{\phivec}{60};

        \tdplotsetthetaplanecoords{\phivec}
        \tdplotdrawarc[tdplot_rotated_coords]{(0,0,0)}{0.2}{0}{\thetavec}{anchor=south}{$\boldsymbol{\gamma}$}

        \tdplotsetcoord{O}{0}{0}{0};
        \tdplotsetcoord{P}{\rvec}{\thetavec}{\phivec}

        \pgfmathsetmacro{\nnx}{1.0*sin(\thetavec)*cos(\phivec)}
        \pgfmathsetmacro{\nny}{1.0*sin(\thetavec)*sin(\phivec)}
        \pgfmathsetmacro{\nnz}{0}

        \draw[-{Stealth[length=3mm, width=2mm]},color=blue] (O) -- (\nnx, \nny, 1)
        node[below=15,left=3]{$\bf{\tilde{n}}$};

        \draw[-{Stealth[length=3mm, width=2mm]},color=blue] (0,0,1) -- (\nnx, \nny, 1)
        node[above=7]{$d\bf{n}$};

        \draw[-{Stealth[length=3mm, width=2mm]},color=blue] (0,0,1) -- (0.5,0,1)
        node[above=15,right=5]{$d\bf{r'}$};

        \draw[-{Stealth[length=3mm, width=2mm]},color=blue] (0,0,0) -- (0.5,0,1)
        node[below=15]{$\bf{r''}$};

        \draw[dashed, color=blue] (\nnx,\nny,1) -- (Pxy);
        \draw[dashed, color=blue] (Px) -- (Pxy);
        \draw[dashed, color=blue] (Py) -- (Pxy);

        \draw[dashed, color=blue] (Py) -- (0,\nny,1);
        \draw[dashed, color=blue] (0,\nny,1) -- (\nnx,\nny,1);
        \draw[dashed, color=blue] (0,\nny,1) -- (0,0,1);
        \draw[dashed, color=blue] (Px) -- (\nnx,0,1);
        \draw[dashed, color=blue] (\nnx,0,1) -- (\nnx,\nny,1);


        \end{tikzpicture}
		\caption{3D view showing how $\bf{r'}$ and $\bf{n}$ change for the infinitesimal motion.}
		\label{fig: geom 1}
	\end{subfigure}
    \vspace{8mm}
 
	\begin{subfigure}{0.85\textwidth}
    \tdplotsetmaincoords{0}{90}
        \begin{tikzpicture}[scale=6, tdplot_main_coords]
        
        \draw[thick,->] (0,0,0) -- (1,0,0) node[anchor=north east]{$\bf{v}$};
        \draw[thick,->] (0,0,0) -- (0,1,0) node[anchor=north west]{$\bf{u}$};

        \pgfmathsetmacro{\rvec}{1.0};
        \pgfmathsetmacro{\thetavec}{30};
        \pgfmathsetmacro{\phivec}{60};

        \tdplotsetcoord{O}{0}{0}{0}
        \tdplotsetcoord{P}{\rvec}{\thetavec}{\phivec}

        \tdplotsetcoord{R}{0.5}{90}{0}
              
        \draw[-{Stealth[length=3mm, width=2mm]},color=blue] (O) -- (P)
        node[above=25, left=12]{$d\bf{n}$};

         \fill [blue] (P) circle[radius=0.4pt]node[below=2, right=2]{$\tilde{\bf{n}}$};

        \draw[dashed, color=blue] (P) -- (Pxz);
        \draw[dashed, color=blue] (P) -- (0.5,0,0);
        
        \pgfmathsetmacro{\ax}{0.5}
        \pgfmathsetmacro{\ay}{-0.1}
        \pgfmathsetmacro{\az}{0}

        \pgfmathsetmacro{\nnx}{1.0*sin(\thetavec)*cos(\phivec)}
        \pgfmathsetmacro{\nny}{1.0*sin(\thetavec)*sin(\phivec)}
        \pgfmathsetmacro{\nnz}{0}

        \draw[-{Stealth[length=3mm, width=2mm]},color=blue] (0,0,0) -- (0.5,0,1) node[above=30, right=0]{$d\bf{r'}$};

         \fill [blue] (0,0,0) circle[radius=0.4pt]node[below=2, left=2]{$\bf{r}=\bf{r'}$};

        \fill [blue] (0.5,0,0) circle[radius=0.4pt]node[below=2, left=2]{$\bf{r''}$};
        
        \node[] at (0.24,0.3) {$\boldsymbol{\theta}$};

        \tdplotdrawarc[dashed]{(\nnx,\nny,\nnz)}{0.5}{180+\phivec}{180+2*\phivec}{anchor=north}{};

        \end{tikzpicture}

		\caption{2D overhead view of the figure in (a) showing the rotation angle $\theta$ as $\bf{r}$ sweeps out an arc in the plane while rotating to $\tilde{\bf{r}}$, which is made to coincide with ${\bf{r}}''$.}
		\label{fig: geom 2}
	\end{subfigure}

	\caption{Geometry for the derivation of the differential equation.}
	\label{fig: geom 1 and geom 2}
\end{figure}

Instead, suppose now that ${\pmb{\omega }}$ is perpendicular to ${\bf{F}}$.  In this case, we will show that we only have to “tilt” ${\bf{F}}$ a little bit, say by angle $\gamma$, without changing its length. That is, we can make this work with ${\bf{\tilde n}} = {\bf{n}} + d{\bf{n}},$ while keeping ${\tilde \theta} = \theta.$ Our goal is to find $d{\bf{n}}$ so that
${\bf{\tilde r}} = {\bf{r}}'' + o(t)$ for all ${\bf{r}}$.  Interestingly, we find that we can determine from pictures what $d{\bf{n}}$ must be, by making ${\bf{\tilde r}} = {\bf{r}}'' + o(t)$ for just \textit{one} judiciously chosen ${\bf{r}}$: namely, choose ${\bf{r}} = {\bf{n}}$.  Then ${\bf{r}}' = {\bf{r}} = {\bf{n}},$ since it is on the rotation axis.  Let ${\bf{u}} = \pmb{\omega }/{\left| \pmb{\omega } \right|}$ and ${\bf{v}} = {\bf{u}} \times {\bf{n}}$. Figure \ref{fig: geom 1 and geom 2} shows what happens to the point ${\bf{r}} = {\bf{n}}.$

We have ${\bf{r}}' = {\bf{r}},$ and then the infinitesimal motion moves it to ${\bf{r}}'' = {\bf{r}}' + d{\bf{r}}'$, where $d{\bf{r}}' = {\pmb{\omega }} \times {\bf{r}}'dt$ 
$ = {\bf{v}}\left| {\pmb{\omega }} \right|dt,$  as shown.  This is a known quantity; it is ${\bf{\tilde n}} = {\bf{n}} + d{\bf{n}}$ that is to be determined.  Now ${\bf{r}}$ is to be rotated by angle $\theta$ about axis $\tilde {\bf{n}},$ which takes place in a plane perpendicular to $\tilde {\bf{n}},$ but since $1-\cos \gamma = O(\gamma ^2) = O((|{{\bf{dn}}}|)^2),$ it is the same as the plane containing the point ${\bf{n}}$ and perpendicular to ${\bf{n}},$ ignoring second-order terms.  Figure ~\ref{fig: geom 2} shows the slice of that from the first figure, shown as coinciding with the plane through ${\bf{n}}$ and perpendicular to $\tilde {\bf{n}}.$  The arc shows ${\bf{r}}$ rotated by $\theta$ about  $\tilde {\bf{n}},$ with  $\tilde{\bf{r}}$ winding up at the same place as ${\bf{r}}'',$ because  $\tilde{\bf{n}}$ is chosen correctly.  Geometry in Figure ~\ref{fig: geom 2} gives $d{\bf{n}} = {\bf{u}}\left| {{\bf{dn}}} \right|\cos (\theta /2) + {\bf{v}}\left| {{\bf{dn}}} \right|\sin (\theta /2),$ and $\left| {{\bf{dn}}} \right| = {\left| {\pmb{\omega }} \right|dt}/(2\sin (\theta /2))$, so
\[d{\bf{n}} = {\bf{u}}\left| {\pmb{\omega }} \right|dt\cot (\theta /2)/2 + {\bf{v}}\left| {\pmb{\omega }} \right|dt/2 = {\pmb{\omega  }}dt\cot (\theta /2)/2 + {\pmb{\omega }} \times {\bf{n}}dt/2.\]
Tracking just a \textit{single} ${\bf{r}}$ has determined what ${\bf{\tilde n}}$ must be, for the case when ${\pmb{\omega }}$ is perpendicular to ${\bf{F}}$. And it shows the geometry that leads to the appearance of $\theta /2$, which is what ultimately causes the spinor-like behavior of ${\bf{n}}(t)$ in some cases, as will be shown.  In \cite{mtw} p. 1137, there is a geometric argument for why half-angles appear in the context of the quaternion representation of finite rotations.  But our setting is completely non-abstract, and the half-angle already shows up in the simple infinitesimal picture.
 
To verify this works for any choice of ${\bf{r}}$, it is enough to verify it for one other ${\bf{r}}$ not collinear with ${\bf{n}}$ and zero, because it is a rigid body. One can use Rodrigues' formula (\ref{eqn: Rod formula 1}) and differential approximation to verify it, with say ${{\bf{r}}}={{\pmb{\omega}}}$. We omit showing that.

Since $\theta $ did not change when ${\pmb{\omega }}$ is perpendicular to ${\bf{F}}$,\\ $d\boldsymbol{F} = d{\bf{n}}f(\theta) = {\pmb{\omega }}f(\theta )\cot(\theta /2 )dt/2 + {\pmb{\omega }} \times \boldsymbol{F}dt/2$ in this case.

To find the total $d\boldsymbol{F}$, since the motions are infinitesimal, we can add together the cases for ${\pmb{\omega }}$ parallel and perpendicular to $\boldsymbol{F}$. In general, write  ${\pmb{\omega }} = {{\pmb{\omega }}_\parallel } + {{\pmb{\omega }}_ \bot } = {\bf{n}}\left( {{\pmb{\omega }} \cdot {\bf{n}}} \right) + \left( {{\pmb{\omega }} - {\bf{n}}\left( {{\pmb{\omega }} \cdot {\bf{n}}} \right)} \right)$. Writing it as a differential equation, $\boldsymbol{\dot{F}} = d\boldsymbol{F}/dt = {\bf{n}}\left( {{\pmb{\omega }} \cdot {\bf{n}}} \right)f'(\theta) + \left( {{\pmb{\omega }} - {\bf{n}}\left( {{\pmb{\omega }} \cdot {\bf{n}}} \right)} \right)f(\theta )\cot(\theta/2)/2 + {\pmb{\omega }} \times \boldsymbol{F}/2$. Rearranging terms, we find our first form of an Euler vector differential equation. 

\subsection{Generalized Euler vector equation.}
\begin{equation}
\label{eqn: gen Euler vec diffeq}
\boldsymbol{\dot{F}}= {\pmb{\omega }}f'(\theta ) - \left\{{\pmb{\omega }}{f^2(\theta)} - {\bf{F}}\left( {{\pmb{\omega }} \cdot {\bf{F}}} \right)\right\}\frac{f'(\theta)-f(\theta)\cot(\theta/2)/2}{f^2(\theta)} + {\pmb{\omega }} \times \boldsymbol{F}/2.
\end{equation}

From assumptions on $f$, $\left\{f'(\theta)-(\theta)\cot(\theta/2)/2\right\}/
{{f^2}(\theta)}$ converges to\\$\left\{f'(0)/12 + 2{f^{(3)}}(0)/3\right\}/f'(0)^2$, as $\theta \rightarrow 0$. Now $f'(0) > 0$, and $f(\theta)$ is invertible on an interval where $f'(\theta) > 0$, and $f$ is odd, so $|f(\theta)| = |{\bf{F}}|$ determines $|\theta|$ for $|{\bf{F}}| \le K$, for some $K>0$.  The r.h.s of (\ref{eqn: gen Euler vec diffeq})  is the same if $\theta \rightarrow -\theta$, so it can be written as a function of $\bf{F}$, for $|{\bf{F}}| \le K$. It follows that
(\ref{eqn: gen Euler vec diffeq}), with initial condition ${\bf{F}}(t_0)= \bf{0}$, is solvable for $t_0 \le t < t_0 + \delta$ for some $\delta > 0$. There is not a problem with the differential equation for $\left| \boldsymbol{F}({t_0}) \right| = 0$. In fact, that is the motivating case, starting from the body system aligned with the fixed system, and growing out of the infinitesimal beginning motion.  Suppose $t_0 = 0$ for convenience. The differential equation near zero is simply $\boldsymbol{\dot{F}} = {\pmb{\omega }}f'(0) + O\left({{\left| \boldsymbol{F} \right|}} \right)$, so then $\boldsymbol{F}(t) = f'(0)\int_0^t {{\pmb{\omega }}(\tau )d\tau }  + O(t^2)$, for $0 \le t \ll 1.$ This is as expected, since infinitesimal rotations just add (integrate) to produce the net rotation vector.
\subsection{Euler axis and angle as separate variables.}
Consider the behavior of $\bf{n}$ and $\theta$.  Suppose ${\bf{F}}(t)$ satisfies (\ref{eqn: gen Euler vec diffeq}) on some interval $t_0 \le t < t_1$, and that $0 < |{\bf{F}}(t)| \le K$ for $t_0 < t < t_1$, so that ${\bf{F}}(t) = {\bf{n}}(t)f(\theta(t))$ uniquely determines ${\bf{n}}(t)$ and $\theta (t) > 0$ as continuous functions on $t_0 < t < t_1$. (Or, the sign of $\theta (t)$ and ${\bf{n}}(t)$ could both be flipped, but we don't need to discuss negative $\theta$ until a later section). If $|{\bf{F}}(t_0)| = 0$, then $\theta(t_0) = 0$ and ${\bf{n}}(t_0)$ is undefined. Now $f^2(\theta) = {\bf{F}} \cdot {\bf{F}}$, so $2f(\theta)f'(\theta){\dot{\theta}} = 2{\bf{F}} \cdot \dot{{\bf{F}}}$, and from (\ref{eqn: gen Euler vec diffeq}), ${\bf{F}} \cdot \dot{{\bf{F}}} = {\pmb{\omega}} \cdot {\bf{F}}f'(\theta)$. Neither $f(\theta)$ or $f'(\theta)$ is zero, so $\dot{\theta} = {\pmb{\omega}} \cdot {\bf{n}}$. Also, differentiating  ${\bf{F}} = {\bf{n}}f(\theta)$ gives ${\dot{\bf{F}}} = {\dot{\bf{n}}}f(\theta) + {\bf{n}}f'(\theta)\dot{\theta}$. This may be solved for ${\dot{\bf{n}}}$, using (\ref{eqn: gen Euler vec diffeq}) and $\dot{\theta} = {\pmb{\omega}} \cdot {\bf{n}}$.  For future reference, we record these, for $t_0 < t < t_1$: 
\begin{equation}
\begin{aligned}
\label{eqn: joint n theta diffeq}
\dot \theta  &= {\pmb{\omega }} \cdot {\bf{n}} \\
{\bf{\dot n}}  &= ({\pmb{\omega }} - {\bf{n}}({\pmb{\omega }} \cdot {\bf{n}}))\cot(\theta /2 )/2  + {\pmb{\omega }} \times {\bf{n}}\,/2 
\end{aligned} 
\end{equation}

The fact that $\boldsymbol{F}(t)$ behaves nicely as $t \rightarrow 0^+$ for any continuous $\pmb{\omega}$, when $|{\bf{F}}(0)| = 0$, does not imply that ${\bf{n}}(t)$ does. Here is a pathological example.
\begin{exa}
 Let
\begin{equation}
\label{eqn: pathological example}
{\pmb{\omega}}(t) = 3t^2 \left( {\bf{i}}\sin(1/t)+{\bf{j}}\cos(1/t)\right) + t \left( -{\bf{i}}\cos(1/t)+{\bf{j}}\sin(1/t)\right),t>0
\end{equation}
with ${\pmb{\omega}}(0) = {\bf{0}}$.
Then $\int_0^t {{\pmb{\omega }}(\tau )d\tau } = t^3 \left( {\bf{i}}\sin(1/t)+{\bf{j}}\cos(1/t)\right)$.  From this and the differential equation it can be shown that ${\bf{n}}(t) = {\bf{i}}\sin(1/t)+{\bf{j}}\cos(1/t) + O(t)$, which does not converge as $t\rightarrow 0^+$; it spins wildly. Later we'll show that with mild regularity conditions on $\pmb{\omega}$, this behavior cannot happen, and $\bf{n}$ approaches a limit as $t\rightarrow 0^+$, and the same for $\dot{\bf{n}}$, with proper conditions on $\pmb{\omega}$.
\end{exa}
This example shows that although coding the rotation as an Euler vector works well all the way down to the infinitesimal case, thinking of the Euler rotation as described by axis and angle separately can break down in the infinitesimal limit.  When the rotation angle is near zero, the points in the rigid body barely move, but the directions of the axis - that is, the ${\bf{n}}$'s - can be changing wildly, if $\pmb{\omega}$ is allowed to be an arbitrary continuous function.

We only look at three examples of generalized Euler vector representations.

\subsection{Euler vector equation.} Take $f(\theta)=\theta$, and refer to $\bf{F}$ as $\bf{E}$.  Then the differential equation, written entirely in terms of $\bf{E}$, becomes
\begin{equation}
\label{eqn: Euler diff eq rep 1}
\boldsymbol{\dot{E}} = {\pmb{\omega }} - \left\{{\pmb{\omega }}{| \boldsymbol{E}|^2} - \boldsymbol{E} \left( {{\pmb{\omega }} \cdot \boldsymbol{E}} \right)\right\} \left( {\frac{{1 - g\left( {\left| \boldsymbol{E} \right|} \right)}}{{{{\left| \boldsymbol{E} \right|}^2}}}} \right) + {\pmb{\omega }} \times \boldsymbol{E}/2,
\end{equation}
where $g(\theta ): = (\theta /2)\cot (\theta /2) = 1 - {\theta ^2}/12 + O(\theta ^4).$ Now $(1-g(\theta))/\theta^2$ is real analytic except at non-zero integer multiples of $2\pi$. It follows from standard facts about differential equations that if $\left| \boldsymbol{E}({t_0}) \right| \ne 2k\pi$ for any positive integer $k$, a solution exists and can be continued for $t > {t_0}$ until/unless $\left| {\boldsymbol{E}(t)} \right|$ approaches $2\pi k$ in finite time, for some positive integer $k$. As mentioned above, $\left| \boldsymbol{E}({t_0}) \right| = 0$ is not a problem. 

Later we'll show that with regularity conditions on $\pmb{\omega}$, the solution for ${\bf{E}}$, as well as $\bf{n}$ and $\theta$ (and their derivatives) may be continued continuously for all $t$; we'll need to let $\theta$ go below zero as well as above $2\pi$. This will accomplish our goal of describing a continuous evolution of the Euler axis and angle, that is, $\bf{n}$ and $\theta$, for all $t$. In many interesting cases, $\theta$ stays between zero and $2\pi$ anyway.

\subsection{Modified Gibbs vector equation, and quaternions.} A glance at (\ref{eqn: gen Euler vec diffeq}) shows that the equation is by far the simplest if we take $f(\theta)=\sin(\theta /2)$, because the middle term disappears entirely. The vector ${\bf{M}} = {\bf{n}}\sin(\theta/2)$ is called the \textit{modified Gibbs representation} in \cite{kruglov}, motivated in a different way. The differential equation becomes
\begin{equation}
\label{eqn: Modified Gibbs vec diffeq}
{\bf{\dot{M}}} = {\pmb{\omega }}\cos(\theta /2)/2 + {\pmb{\omega }} \times \boldsymbol{M}/2.
\end{equation}
This $f(\theta)$ is invertible on the interval $|\theta| \le \pi$, and $|\sin(\theta/2)| = |{\bf{M}}|$ determines $\cos(\theta/2) = \sqrt{1 - |{\bf{M}}|^2}$, which allows (\ref{eqn: Modified Gibbs vec diffeq}) to be written entirely in terms of $\bf{M}$; but the sign of $\cos(\theta /2)$ cannot be determined from $\bf{M}$ if $\theta$ is allowed to be greater than $\pi$. That problem can be eliminated by introducing a fourth variable, letting $M_0 = \cos(\theta /2)$, so $\dot{M_0} = -\dot{\theta}\sin(\theta /2)/2 = -\pmb{\omega} \cdot {\bf{n}}\sin(\theta /2)/2 = -\pmb{\omega} \cdot {\bf{M}}/2$, using (\ref{eqn: joint n theta diffeq}).  We now have a differential equation in 4 variables $(M_0,{\bf{M}}) = (M_0,M_1,M_2,M_3)$:
\begin{equation}
\label{eqn: quaternion diffeq}
{\bf{\dot{M}}} = {\pmb{\omega }}M_0 /2 + {\pmb{\omega }} \times \boldsymbol{M}/2;
\dot{M_0} = -{\pmb{\omega}} \cdot {\bf{M}}/2.
\end{equation}
This is in fact a well-known quaternion differential equation \cite{chou}. Let $Q=(M_0,{\bf{M}})=$\\$(M_0,M_1,M_2,M_3)$. This is a unit quaternion, since $\sum_{i=1}^{4}{M_i}^2 = 1$.  Equation (\ref{eqn: quaternion diffeq}) can be written as ${\dot{Q}} = \omega \otimes Q /2$, where $\omega = (0,{\pmb{\omega}})$, a quaternion with zero "real" part, and $\otimes$ is quaternion multiplication. We arrived at the equation in a different way, using only infinitesimal rigid body motions in real three-dimensional space.

For any initial condition $(M_0(t_0), {\bf{M}}(t_0))$ satisfying the normalization condition $M_0(t_0)^2 + {\bf{M}}(t_0) \cdot {\bf{M}}(t_0)  = 1$, equation (\ref{eqn: quaternion diffeq}) is solvable for all $t$, and the differential equation shows that $(d/dt)\{M_0(t)^2 + {\bf{M}}(t) \cdot {\bf{M}}(t)\} = 0$, so the normalization condition remains true.  If $|M_0(t)| < 1$, we have ${\bf{M}}(t) \ne {\bf{0}}$, so a unique $0 < \theta(t) < 2\pi$ is determined from $\cos(\theta(t)/2 = M_0(t)$, and ${\bf{n}}(t) = {\bf{M}}(t)/\sin(\theta(t)/2)$. This makes  ${\bf{n}}$ differentiable on any interval where $|M_0(t)| < 1$, and a similar argument to the one given to prove (\ref{eqn: joint n theta diffeq}) shows that (\ref{eqn: joint n theta diffeq}) holds here as well. If $|M_0(t_1)| = 1$ for some $t_1$, the same example as above shows that ${\bf{n}}(t_1)$ cannot necessarily be defined to make ${\bf{n}}$ continuous at $t_1$, without some conditions (to be discussed later) on ${\pmb{\omega}}$.

\subsection{Gibbs vector equation.} The representation ${\bf{G}} = {\bf{n}}\tan(\theta/2)$ is called variously the Gibbs vector \cite{kruglov}, or the Rodrigues vector \cite{wikiRotationformalisms}. Differential equation (\ref{eqn: gen Euler vec diffeq}) becomes
\begin{equation}
\label{eqn: Gibbs vec diffeq}
{\bf{\dot{G}}} = {\pmb{\omega }}/2 + \boldsymbol{G} \left( {{\pmb{\omega }} \cdot \boldsymbol{G}} \right)/2 + {\pmb{\omega }} \times \boldsymbol{G}/2.
\end{equation}
This is a three-variable equation which is simpler than the other three-variable equations, which is interesting.  But ${\bf{G}}$ blows up as $\theta$ approaches $\pi$, and there is no fix for this.\\

From here on, we will mainly consider the Euler vector equation (\ref{eqn: Euler diff eq rep 1}) for ${\bf{E}}$, and the 4-variable (quaternion) version (\ref{eqn: quaternion diffeq}) of the Modified Gibbs representation ${\bf{M}}$. The Euler vector codes arbitrarily large rotation angles $\theta$, which will be useful when we discuss having a continuously evolving axis and angle; but the solution to (\ref{eqn: Euler diff eq rep 1}) cannot necessarily be continued for all time without some mild restrictions on ${\pmb{\omega}}$. The quaternion solution exists for all time with no restrictions, but does not distinguish between angles that differ by multiples of $2\pi$.

The following theorem formally states and proves what was indicated above by a geometric argument, for these two cases.
\begin{theorem}
\label{thm: main}
  Let ${\pmb{\omega }}$ be the angular velocity function of a rigid body for which the origin of the body system is constrained to be at the origin of the space system.
  
  Let ${\bf{E}}(t)$ solve equation (\ref{eqn: Euler diff eq rep 1}) for ${t_0} \le t < {t_1}$, with $\left| {\bf{E}}(t) \right| \ne 2k\pi$ for any positive integer $k$. Let $\theta(t) = \left| {\bf{E}}(t) \right|$, and ${\bf{n}}(t) = {\bf{E}}(t) / \theta(t)$ if $\theta(t) \ne 0$.

  Or, let $({M_0(t),\bf{M}}(t))$ solve equation (\ref{eqn: quaternion diffeq}) for ${t_0} \le t < t_1$, with $M_0(t_0)^2 + {\bf{M}}(t_0) \cdot {\bf{M}}(t_0)  = 1$. Let $\theta(t)= 2{\cos^{ - 1}}(M_0(t))$, using any branch of the $\arccos$ function, and ${\bf{n}}(t) = {\bf{M}}(t) /\sin(\theta(t)/2)$ if $\left| {\bf{M}}(t) \right| \ne 0$.
  
  If the position of the body at time ${t_0}$ is $Rot_{{\bf{n}}(t_0),\theta (t_0)}$, then the position at time $t$ is $Rot_{{\bf{n}}(t),\theta (t)}$, for $t_0 \le t < t_1$.
  \end{theorem}
\begin{proof} We describe a rigorous proof, one which also does not assume Euler's finite rotation theorem.
For a point in the body, let ${{\bf{r}}_b}$ be its position in the body coordinate system. If $\theta(t) = 2\pi m$ for some integer $m$ (in which case ${\bf{n}}(t)$ is not defined), define ${\bf{r}}(t) = {\bf{r}}_b$. Otherwise, define ${\bf{r}}(t) = {{\bf{r}}_b}\cos \theta  + {\bf{n}}({\bf{n}} \cdot {{\bf{r}}_b})(1 - \cos \theta ) + ({\bf{n}} \times {{\bf{r}}_b})\sin \theta$; by equation (\ref{eqn: Rod formula 1}), this is the rotation of ${{\bf{r}}_b}$ specified by axis-angle $({\bf{n}}(t),\theta(t))$.  If we show that ${\bf{\dot r}}(t) = {\pmb{\omega }}(t) \times {\bf{r}}(t)$, ${t_0} \le t < {t_1}$, the theorem will be proved, since the solution to the differential equation (\ref{eqn: Ang velocity}) uniquely determines the path of a point in the body, given its starting position.

If $\theta(t) \ne 2\pi m$ for any integer $m$, ${\bf{n}}$ is defined in a nbhd of $t$.  Compute
${\bf{\dot r}}(t) = ( - {{\bf{r}}_b} + {\bf{n}}({\bf{n}} \cdot {{\bf{r}}_b}))\dot \theta \sin \theta  + ({\bf{\dot n}}({\bf{n}} \cdot {{\bf{r}}_b}) + {\bf{n}}({\bf{\dot n}} \cdot {{\bf{r}}_b}))(1 - \cos \theta ) + ({\bf{\dot n}} \times {{\bf{r}}_b})\sin \theta  + ({\bf{n}} \times {{\bf{r}}_b})\dot \theta \cos \theta$. Use (\ref{eqn: joint n theta diffeq}) to express ${\bf{\dot n}}$ and $\dot \theta$ in this expression. Then with some amount of algebra, using a vector triple product identity and the trig identity $\cot(\theta /2) = \sin(\theta)/(1- \cos(\theta))$, we get\\
${\bf{\dot r}} = \left( {{\bf{n}}\left( {({\pmb{\omega }} \times {\bf{n}}) \cdot {{\bf{r}}_b}} \right) - {\bf{n}} \times {{\bf{r}}_b}({\pmb{\omega }} \cdot {\bf{n}})
 +{\pmb{\omega }} \times {\bf{n}}({\bf{n}} \cdot {{\bf{r}}_b})} \right)(1 - \cos \theta )/2\\
+\left( {{\bf{n}}({\pmb{\omega }} \cdot {{\bf{r}}_b}) - {{\bf{r}}_b}({\pmb{\omega }} \cdot {\bf{n}})} \right)\sin \theta  + {\pmb{\omega }} \times {{\bf{r}}_b}(1 + \cos \theta )/2$. But\\
${\pmb{\omega }} \times {\bf{r}} = {\pmb{\omega }} \times {{\bf{r}}_b}\cos \theta  + {\pmb{\omega }} \times {\bf{n}}({\bf{n}} \cdot {{\bf{r}}_b})(1 - \cos \theta ) + \left( {{\bf{n}}({\pmb{\omega }} \cdot {{\bf{r}}_b}) - {{\bf{r}}_b}({\pmb{\omega }} \cdot {\bf{n}})} \right)\sin \theta $; we need to show this is equal to ${\bf{\dot r}}$.   Subtracting from ${\bf{\dot r}}$ and rearranging and writing $1={\bf{n}}\cdot{\bf{n}}$ to facilitate using an identity, we get ${\bf{\dot r}} - {\pmb{\omega }} \times {\bf{r}} ={ \bf{Q}}(1 - \cos \theta )/2$, where \\
$\bf{Q} =  {{\bf{n}}\,\left( {{{\bf{r}}_b} \cdot ({\pmb{\omega }} \times {\bf{n}})} \right) - {\pmb{\omega }} \times {\bf{n}}\,({\bf{n}} \cdot \,{{\bf{r}}_b}) - {\bf{n}} \times \,{{\bf{r}}_b}({\pmb{\omega }} \cdot {\bf{n}}) - \,{{\bf{r}}_b} \times {\pmb{\omega }}({\bf{n}} \cdot {\bf{n}})}$. The following identity of Gibbs from 1901 (which does not seem to appear in current textbooks) shows $\bf{Q}$ is identically zero.

\begin{lemma}
[Vector identity] (\cite{gibbs_id} p. 77, also \cite{wikigibbs_id}). For ${\bf{A,B,C,D}}$ in $\mathbb{R}^3$,\\
${\bf{D}}\left( {{\bf{A}} \cdot ({\bf{B}} \times {\bf{C}})} \right) = {\bf{B}} \times {\bf{C}}({\bf{A}} \cdot {\bf{D}}) + {\bf{C}} \times {\bf{A}}({\bf{B}} \cdot {\bf{D}}) + {\bf{A}} \times {\bf{B}}({\bf{C}} \cdot {\bf{D}})$.
\end{lemma}
Letting ${\bf{A}} = {{\bf{r}}_b},{\bf{B}} = {\pmb{\omega }},{\bf{C}} = {\bf{D}} = {\bf{n}}$ in the identity shows ${\bf{Q}} = {\bf{0}}$.\\

There remains the case $\theta (t) = 2\pi m$ for integer $m$, so ${\bf{r}}(t)={\bf{r}}_b$. Then ${\bf{r}}(t')-{\bf{r}}(t)\\= \{-{{\bf{r}}_b} + {\bf{n}}(t')({\bf{n}}(t') \cdot {{\bf{r}}_b})\}(1 - \cos\theta(t'))  + ({\bf{n}}(t') \times {{\bf{r}}_b})\sin \theta(t')\\= ({\bf{n}}(t') \times {{\bf{r}}_b})\theta(t') + O((\theta(t'))^2)$. (If $\cos \theta(t') = 1$, this is zero without ${\bf{n}}(t')$ being defined.) We give the proof for ${\bf{E}}$, the proof for ${\bf{M}}$ being similar. For this case, $\theta (t) = 0$.
From (\ref{eqn: Euler diff eq rep 1}), and using ${\bf{E}}(t) = {\bf{0}}$, ${\bf{E}}(t') = {\bf{E}}(t) + {\dot{\bf{E}}(t)}h + o(h) = {\pmb{\omega }}(t)h + o(h)$, where $h=t'-t$. Thus $\theta(t') = |{\bf{E}}(t')| = {\pmb{\omega }}(t)|h| + o(h)$. If
$|{\pmb{\omega }}(t)| \ne 0$, ${\bf{n}}(t') = {\bf{E}}(t')/|{\bf{E}}(t')| = {\pmb{\omega }}(t)/|{\pmb{\omega }}(t)| + o(1)$, so ${\bf{r}}(t')-{\bf{r}}(t) = \{{\pmb{\omega }}(t)/|{\pmb{\omega }}(t)| + o(1)\} \times {\bf{r}}_b(|{\pmb{\omega }}(t)|h| + o(h)) = ({\pmb{\omega }}(t) \times {\bf{r}}_b) h + o(h)$. So $\{{\bf{r}}(t')-{\bf{r}}(t)\}/h \rightarrow {\pmb{\omega }}(t) \times {\bf{r}}_b$ as $h \rightarrow 0$, and ${\dot{\bf{r}}}(t) = {\pmb{\omega }}(t) \times {\bf{r}}_b$.  If
$|{\pmb{\omega }}(t)| = 0$, ${\bf{r}}(t')-{\bf{r}}(t) = o(h)$, and 
${\dot{\bf{r}}}(t) = 0 = {\pmb{\omega }}(t) \times {\bf{r}}_b$.
\end{proof}

Theorem \ref{thm: main} accomplishes what initially motivated us: to use only infinitesimal motions to generate all future positions of the body, by means of a differential equation, and in the process showing that finite motions are also rotations, giving yet another proof of Euler's rotation theorem. The version with ${\bf{M}}$ gives an axis and angle for the position for all time, though they do not necessarily evolve continuously for all time; example (\ref{eqn: pathological example}) shows that such is not always possible.  The version with ${\bf{E}}$ also could be used to verify all future positions are rotations, even though the solution for ${\bf{E}}$ cannot necessarily be continued for all time, if $|{\bf{E}}(t)|$ approaches a positive multiple of $2\pi$ in finite time. If $|{\bf{E}}(t)|$ reaches say $3\pi/2$ at time $t_1$, restart with initial condition $-{\bf{E}}(t_1)/3$, which amounts to flipping the axis and starting over with rotation angle $\pi/2$, corresponding to the same physical position. This avoids dealing with the boundary issues.

In Section \ref{sec: boundary}, we'll show that that, in fact, it is possible to have continuous evolution of axis and angle, and also continue the solution for ${\bf{E}}$ for all time, with mild regularity assumptions on ${\pmb{\omega }}$; for example, if it is never zero. But first, in Section \ref{sec: spinor}, we present some interesting cases where the angle stays between $0$ and $2 \pi$ automatically, and ${{\bf{n}}(t)} = {{\hat{\boldsymbol{E}}}(t)} = \boldsymbol{E}(t)/\left| {\boldsymbol{E}(t)} \right|$  exhibits spinor-like behavior.

\section{\textbf{Simple examples. The Euler Spinor.}}
\label{sec: spinor}

Suppose that from time $t_0$ forward, ${\pmb{\omega }}$ is constant and not zero.  If ${\bf{E}}$ is parallel to ${\pmb{\omega }}$ at that time, then it stays parallel, and only the length of ${\bf{E}}$ changes, going forward.  But if ${\bf{E}}$ is not parallel to ${\pmb{\omega }}$, we will show ${\bf{\hat E}}$  moves in a plane circle, with spinor-like behavior, tracking the rays of a simply parameterized ellipse, with twice the period of the rotating body, and changing signs as the body goes through one revolution.  Similar to the discussion in \cite{mtw} p. 1149, we think of the Euler vector solution, being defined by its relation to the surrounding fixed system, as manifesting "entanglement" with the surrounding after only one revolution of the body, even though the body itself is restored to its original position.  We call ${\bf{\hat E}}$ the \textit{Euler Spinor}.  We first discovered this result by numerical experiments, as seen in the motivating Figures \ref{fig:Euler vec const w trajectory} and \ref{fig:Euler vec const w ind variables} and the corresponding animation. Next we give an analytical solution to describe this behavior.  

 For notational convenience, we state the result for $t_0=0$ and $\left| {\pmb{\omega }} \right|=1$. See Figure \ref{fig: euler_spinor_tikz}.


\begin{figure}
\centering
  
\begin{tikzpicture}
	\node[ellipse, dashed,
	draw = black,
	text = blue,
	minimum width = 10cm, 
	minimum height = 4cm] (e) at (0,0) {};

 \node[circle,
    draw,
    minimum size =4cm] (c) at (0,0){};

    \draw[line width=1pt,black,-stealth](0,0) -- (2,0) node[anchor= west]{$\boldsymbol{e_2}$};
    
    \draw[line width=1pt,black,-stealth](0,0) -- (0,-2) node[anchor=north ]{$\boldsymbol{e_1}$};

      \draw[line width=1pt,red,-stealth](0,0)--(2*0.2588, -2*0.9659) node[anchor=north ]{$\boldsymbol{{\hat{E}}_0}$};


      \draw[line width=1pt,blue,-stealth](0,0)--(2*0.866, -2*0.5) node[above=3,right=0]{$\boldsymbol{\hat{E}}(t)$};

      \draw[line width=1pt,blue,-stealth,dashed](2*0.866, -2*0.5)--(3.27*0.866, -3.27*0.5) node[below=3,right=1]{$\boldsymbol{p}(t)$};

    \node[] () [below = 2.5cm] at (0,0) {};

\end{tikzpicture}

	\caption{Euler spinor dynamics. The normalized Euler vector ${\bf{\hat E}}(t)$ moves in a plane unit circle inscribed in an ellipse, guided by the ray to a point ${\bf{p}}(t)$ moving on the ellipse. The speed is maximum when ${\bf{\hat E}}(t)$ is along ${\bf{e}}_1$, the direction of the semi-minor axis of the ellipse, and the speed is minimum along ${\bf{e}}_2$, the semi-major axis.}
	\label{fig: euler_spinor_tikz}
\end{figure}
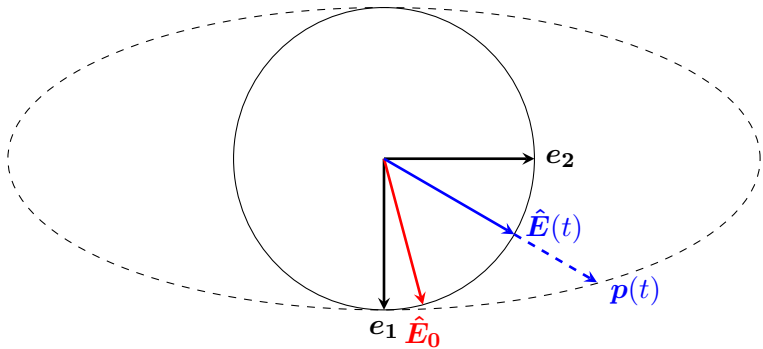

\begin{theorem}
\label{thm: spinor}
Suppose ${\pmb{\omega }}(t) = {\pmb{\omega }}$ is constant for $t \ge {t_0}$, and $\left| {\pmb{\omega }} \right|=1$.  Let $0 < {\theta _0} < 2\pi $, and let ${{\bf{n}}_0}$ be a unit vector such that $\left| {{\pmb{\omega }} \cdot {{\bf{n}}_0}} \right| < 1$; that is, ${{\bf{n}}_0}$ is not parallel to ${\pmb{\omega }}$.  Then there is ${\bf{E}}(t)$, for all $t \ge {t_0}$, that solves equation (\ref{eqn: Euler diff eq rep 1}) with initial condition ${\bf{E}}({t_0}) = {{\bf{n}}_0}{\theta _0}$. Let ${{\bf{n}}(t)} = {\bf{\hat E}}(t)$, so ${{\bf{n}}_0} = {{\bf{\hat E}}_0}$, and $\theta(t) = \left| {{\bf{E}}(t)} \right|$, so that ${\bf{E}}(t) = {\bf{n}}(t)\theta(t)$. Both $\theta $ and ${\bf{n}}$ are $4\pi$  periodic.  We have
\begin{flalign*}
\theta \left( {t + 2\pi} \right) = 2\pi  - \theta \left( t \right),\text{and } {\bf{n}}\left( {t + 2\pi } \right) =  - {\bf{n}}(t),
\end{flalign*}
so that ${\bf{n}}$ changes sign when the body goes through one complete revolution.   Letting ${{\mathop{\rm Cos}\nolimits} ^{ - 1}}$ be the principal branch of the $\arccos $ function, there are constants $a$ and $b$, with $0 < \left| {{\pmb{\omega }} \cdot {{\bf{n}}_0}} \right| <a < 1$, such that 
\begin{flalign*}
\theta (t) = 2{{\mathop{\rm Cos}\nolimits} ^{ - 1}}\left( {a \cos \left( {t/2\,\,\, + \,\,\,b} \right)} \right),
\end{flalign*}
so $\theta $ oscillates between ${\theta _{\min }} = 2{{\mathop{\rm Cos}\nolimits} ^{ - 1}}a$  and ${\theta _{\max }} = 2{{\mathop{\rm Cos}\nolimits} ^{ - 1}}( - a) = 2\pi  - {\theta _{\min }}.$ We have $1-{a^2} =  \left( 1 - \left( {\pmb{\omega }} \cdot {{\bf{n}}_0} \right)^2 \right){{\sin }^2}({\theta _0}/2)$, $b = {{\mathop{\rm Cos}\nolimits} ^{ - 1}}\left( {a^{-1}} \cos ({\theta _0}/2) \right)$, for computing $a$ and $b$.

There are orthogonal unit vectors ${{\bf{e}}_1},{{\bf{e}}_2}$ such that ${\bf{n}}(t) = {\bf{p}}(t)/\left| {{\bf{p}}(t)} \right|$ for $t \ge 0$, where
\begin{flalign*}
 {\bf{p}}(t) = {{\bf{e}}_1}\cos \left( {t/2\,\,\, + \,\,\,b} \right) + {{\bf{e}}_2}\left( {1 - {a^2}}\right)^{-1/2}\sin \left( {t/2\,\,\, + \,\,\,b} \right).
\end{flalign*}
Thus ${\bf{n}}(t)$ moves in a plane unit circle inscribed in an ellipse, guided by the ray to a point moving on the ellipse, as in Figure \ref{fig: euler_spinor_tikz}.  The speed $\left| {{\bf{\dot n}}(t)} \right|$ is maximum when ${\bf{n}}(t) =  \pm {{\bf{e}}_1}$, the direction of the semi-minor axis of the ellipse.  It is slowest when ${\bf{n}}(t) =  \pm {{\bf{e}}_2}$ .  In terms of the initial conditions, a formula for these basis vectors for the plane is
${{\bf{e}}_1} = {{\bf{n}}_0}\sqrt {1 - {k^2}}  - {\bf{u}}k$, ${{\bf{e}}_2} = {{\bf{n}}_0}k + {\bf{u}}\sqrt {1 - {k^2}} $, where ${\bf{u}} = {\bf{\dot n}}(0)/\left| {\bf{\dot n}}(0) \right|$ is a unit vector orthogonal to ${{\bf{n}}_0}$, and $k = {\pmb{\omega }} \cdot {{\bf{n}}_0}/a$.  And ${\bf{u}} = {\bf{u}}_1\cos ({\theta _0}/2) + {\bf{u}}_2\sin ({\theta _0}/2)$, where ${\bf{u}}_1 = \left({\pmb{\omega }} - {{\bf{n}}_0}({\pmb{\omega }} \cdot {{\bf{n}}_0}\right)/\sqrt {1 - ({\pmb{\omega }} \cdot {{\bf{n}}_0})^2}$
 is the normalization of ${\pmb{\omega }} - pro{j_{{{\bf{n}}_0}}}({\pmb{\omega }})$, and ${\bf{u}}_2 = {\bf{u}}_1 \times {{\bf{n}}_0}$ is the normalization of ${\pmb{\omega }} \times {{\bf{n}}_0}.$  Geometrically, ${\bf{u}}$ is obtained by rotating ${\bf{u}}_1$ by angle ${\theta _0}/2$ towards ${\bf{u}}_2$.
\end{theorem}
\begin{proof}
Since $\pmb{\omega}$ is constant, one could get a second order differential equation involving only $\theta$ with $\bf{n}$ eliminated, and solve that by separation of variables, and then solve for $\bf{n}$.  But since the position at time $t$ is the result of rotation by angle ${\theta}_0$ about axis ${\bf{n}}_0$ followed by rotation by angle $t$ about axis $\pmb{\omega}$, we can instead just use Rodrigues' 1840 formula for the composition of two rotations (see \cite{wikiRotationformalisms}), stated in axis-angle terms: 

Let $R_3 = R_2R_1$, where $R_i$ is rotation by angle ${\alpha}_i$ counterclockwise about unit vector axis ${\bf{A}}_i$.  Then
\begin{flalign}
 \cos({{\alpha}_3}/2)& = \cos({{\alpha}_2}/2)\cos({{\alpha}_1}/2) - {{\bf{A}}_2} \cdot {{\bf{A}}_1}\sin({{\alpha}_2}/2)\sin({{\alpha}_1}/2) \label{eqn: Rod comp angle}  \\
 {{\bf{A}}_3}\sin({{\alpha}_3}/2)& = {{\bf{A}}_2}\sin({{\alpha}_2}/2)\cos({{\alpha}_1}/2) + {{\bf{A}}_1}\cos({{\alpha}_2}/2)\sin({{\alpha}_1}/2) \label{eqn: Rod comp vector}\\
 & + {{\bf{A}}_2} \times {{\bf{A}}_1}\sin({{\alpha}_2}/2)\sin({{\alpha}_1}/2)\nonumber
\end{flalign}
Applied to our problem, the angle equation (\ref{eqn: Rod comp angle}) is\\ $\cos({{\theta}(t)}/2) = \cos(t/2)\cos({{\theta}_0}/2) - {{\pmb{\omega}}} \cdot {{\bf{n}}_0}\sin(t/2)\sin({{\theta}_0}/2).$ In phasor form,\\
$\cos(\theta (t)/2) = \sqrt {{{\cos }^2}({\theta _0}/2) + {{({\pmb{\omega }} \cdot {{\bf{n}}_0})}^2}{{\sin }^2}({\theta _0}/2)} \cos (t/2 + b)$\\$ = \sqrt {1 - \left( {1 - {{({\pmb{\omega }} \cdot {{\bf{n}}_0})}^2}} \right){{\sin }^2}({\theta _0}/2)} \cos (t/2 + b)$$ = a\cos (t/2 + b)$,\\ where $a = \sqrt {1 - \left( {1 - {{({\pmb{\omega }} \cdot {{\bf{n}}_0})}^2}} \right){{\sin }^2}({\theta _0}/2)} $ and $b = {{\mathop{\rm Cos}\nolimits} ^{ - 1}}\left( {{a^{ - 1}}\cos ({\theta _0}/2)} \right).$\\ \\
The vector equation \ref{eqn: Rod comp vector} gives ${\bf{n}}(t)\sin (\theta (t)/2) = {\pmb{\omega }}\sin (t/2)\cos ({\theta _0}/2)$$ + {{\bf{n}}_0}\sin ({\theta _0}/2)\cos (t/2)$$\\ + {\pmb{\omega }} \times {{\bf{n}}_0}\sin (t/2)\sin ({\theta _0}/2)$.  We will put this in the form stated in the theorem, to get our geometric interpretation.  By standard identities it can be written in terms of $\cos (t/2 + b)$ and $\sin (t/2 + b)$ as
${\bf{n}}(t)\sin (\theta (t)/2) = {{\bf{v}}_1}\cos (t/2 + b) + {{\bf{v}}_2}\sin (t/2 + b),$ where\\
${{\bf{v}}_1} = {{\bf{n}}_0}\sin ({\theta _0}/2)\cos b - {\pmb{\omega }}\cos ({\theta _0}/2)\sin b - {\pmb{\omega }} \times {{\bf{n}}_0}\sin ({\theta _0}/2)\sin b,$ and\\
${{\bf{v}}_2} = {{\bf{n}}_0}\sin ({\theta _0}/2)\sin b + {\pmb{\omega }}\cos ({\theta _0}/2)\cos b + {\pmb{\omega }} \times {{\bf{n}}_0}\sin ({\theta _0}/2)\cos b.$   By definition, $a\cos b = \cos ({\theta _0}/2)$, and $1 - {a^2}$ = ${\sin ^2}({\theta _0}/2)\left( {1 - {{({\pmb{\omega }} \cdot {{\bf{n}}_0})}^2}} \right)$, from which $\sin b = ({\pmb{\omega }} \cdot {{\bf{n}}_0}/a)\sin ({\theta _0}/2).$   With foresight, write in terms of three mutually orthogonal vectors ${{\bf{n}}_0},$ ${\pmb{\omega }} - {{\bf{n}}_0}({\pmb{\omega }} \cdot {\bf{n}}),$ and ${\pmb{\omega }} \times {{\bf{n}}_0}.$ We have\\
$a{{\bf{v}}_1}$ $= {{\bf{n}}_0}\left( {1 - {{({\pmb{\omega }} \cdot {{\bf{n}}_0})}^2}} \right)\sin ({\theta _0}/2)\cos ({\theta _0}/2)$$ - \left( {{\pmb{\omega }} - {{\bf{n}}_0}({\pmb{\omega }} \cdot {{\bf{n}}_0})} \right)({\pmb{\omega }} \cdot {{\bf{n}}_0})\cos ({\theta _0}/2)\sin ({\theta _0}/2)$$ -\\ ({\pmb{\omega }} \times {{\bf{n}}_0})({\pmb{\omega }} \cdot {{\bf{n}}_0}){\sin ^2}({\theta _0}/2)$$= \sin ({\theta _0}/2)\sqrt {1 - {{({\pmb{\omega }} \cdot {{\bf{n}}_0})}^2}} \left\{ {{{\bf{n}}_0}\sqrt {1 - {{({\pmb{\omega }} \cdot {{\bf{n}}_0})}^2}} \cos ({\theta _0}/2) - {\bf{u}}({\pmb{\omega }} \cdot {{\bf{n}}_0})} \right\},$\\
 where ${\bf{u}} = {{\bf{u}}_1}\cos ({\theta _0}/2) + {{\bf{u}}_2}\sin ({\theta _0}/2),$ with ${{\bf{u}}_1} = ({{\pmb{\omega }} - {{\bf{n}}_0}({\pmb{\omega }} \cdot {{\bf{n}}_0})}) / {{\sqrt {1 - {{({\pmb{\omega }} \cdot {{\bf{n}}_0})}^2}} }},$ and ${{\bf{u}}_2} = ({{\pmb{\omega }} \times {{\bf{n}}_0}}) / {{\sqrt {1 - {{({\pmb{\omega }} \cdot {{\bf{n}}_0})}^2}} }}$, which are unit vectors.   Similarly,\\
 $a{{\bf{v}}_2} = {{\bf{n}}_0}({\pmb{\omega }} \cdot {{\bf{n}}_0})\left( {{{\sin }^2}({\theta _0}/2) + {{\cos }^2}({\theta _0}/2)} \right)$$+ \left( {{\pmb{\omega }} - {{\bf{n}}_0}({\pmb{\omega }} \cdot {{\bf{n}}_0})} \right){\cos ^2}({\theta _0}/2)$\\
 $+  {\pmb{\omega }} \times {{\bf{n}}_0}\sin ({\theta _0}/2)\cos ({\theta _0}/2)$$ = {{\bf{n}}_0}({\pmb{\omega }} \cdot {{\bf{n}}_0}) + {\bf{u}}\sqrt {1 - {{({\pmb{\omega }} \cdot {{\bf{n}}_0})}^2}} \cos ({\theta _0}/2).$  Equation \ref{eqn: joint n theta diffeq} can be used to show that ${\bf{u}}$ is the normalization of ${\bf{\dot n}}(0)$. It is now apparent that  ${{\bf{v}}_1}$ is orthogonal to  ${{\bf{v}}_2}$, and that $\left| {{{\bf{v}}_1}} \right| = \sin ({\theta _0}/2)\sqrt {1 - {{({\pmb{\omega }} \cdot {{\bf{n}}_0})}^2}} \left| {{{\bf{v}}_2}} \right| = \sqrt {1 - {a^2}} \left| {{{\bf{v}}_2}} \right|.$  And $|a{{\bf{v}}_2}{|^2} = {({\pmb{\omega }} \cdot {{\bf{n}}_0})^2} + \left( {1 - {{({\pmb{\omega }} \cdot {{\bf{n}}_0})}^2}} \right){\cos ^2}({\theta _0}/2)$\\$ = 1 - \left( {1 - {{({\pmb{\omega }} \cdot {{\bf{n}}_0})}^2}} \right){\sin ^2}({\theta _0}/2) = {a^2},$ so $\left| {{{\bf{v}}_2}} \right| = 1,$ and so ${{\bf{v}}_2} = {{\bf{n}}_0}({\pmb{\omega }} \cdot {{\bf{n}}_0}/a) + {\bf{u}}\sqrt {1 - {{({\pmb{\omega }} \cdot {{\bf{n}}_0}/a)}^2}} .$ Similarly, ${{\bf{n}}_0}\sqrt {1 - {{({\pmb{\omega }} \cdot {{\bf{n}}_0})}^2}} \cos ({\theta _0}/2) - {\bf{u}}({\pmb{\omega }} \cdot {{\bf{n}}_0})$ has length $a$, so ${\left( {1 - {a^2}} \right)^{ - 1/2}}{{\bf{v}}_1} = \\ {{\bf{n}}_0}\sqrt {1 - {{({\pmb{\omega }} \cdot {{\bf{n}}_0}/a)}^2}}  - {\bf{u}}({\pmb{\omega }} \cdot {{\bf{n}}_0}/a).$ Thus ${\bf{n}}(t)\sin ({\theta (t)}/2){\left( {1 - {a^2}} \right)^{ - 1/2}}$$ = {{\bf{e}}_1}\cos (t/2 + b) + \\ {{\bf{e}}_2}{\left( {1 - {a^2}} \right)^{ - 1/2}}\sin (t/2 + b) = {\bf{p}}(t),$ and ${\bf{n}}(t)$ is the normalization of this. \\
\end{proof} 

\begin{exa}
Suppose the body starts out from having been rotated 90 degrees about the $x$ axis, and thereafter has angular velocity vector ${\bf{k}}$, so it is spinning about the vertical.  That is, ${\theta _0} = \pi /2$, ${{\bf{n}}_0} = {\bf{i}}$, and ${\pmb{\omega }} = {\bf{k}}$, and let ${t_0} = 0$.  Then in Theorem \ref{thm: spinor}, $a = 1/\sqrt 2 $ , $b = 0$, and $\theta (t) = 2{{\mathop{\rm Cos}\nolimits} ^{ - 1}}\left( {{2^{ - 1/2}}\cos \left( {t/2} \right)} \right)$. Also, one can compute ${{\bf{e}}_1} = {\bf{i}}$ and ${{\bf{e}}_2} = ({\bf{j}} + {\bf{k}})/\sqrt{2}$, so ${\bf{n}}(t) = (\cos(t/2),\sin(t/2),\sin(t/2))/\sqrt{1+{{\sin}^2}(t/2)}.$ This is illustrated above in Figures \ref{fig:Euler vec const w trajectory} and \ref{fig:Euler vec const w ind variables}.   Note ${\bf{n}}$ starts at ${{\bf{i}}}$, picks up some positive 2nd and 3rd components, and then winds up at $ - {{\bf{i}}}$ as $\theta$ goes from $\pi /2$ to $3\pi /2$, as the body makes one complete revolution.   Then ${\bf{n}}$ picks up some negative  2nd and 3rd components, and winds up at ${{\bf{i}}}$ as $\theta $ goes from $3\pi /2$  back to $\pi /2$, as the body makes another complete revolution.  Any example of Theorem \ref{thm: spinor} behaves similarly, except that the offset $b$ is not zero unless ${\pmb{\omega }}$ is perpendicular to ${{\bf{n}}_0}$, so in general $\bf{n}$ does not start out at ${{\bf{e}}_1}$.  See Figure \ref{fig: euler_spinor_tikz}.
\end{exa}

The spinor-like behavior of Theorem \ref{thm: spinor}, in which ${\bf{n}}$ changes sign when the body makes one complete revolution,  does not actually require that ${\pmb{\omega }}$ be constant in magnitude, only in direction.  This is consequence of a simple change of time argument that is analogous to changing to arclength parameter for a curve.

\begin{lemma}
Suppose $\left| {{\pmb{\omega }}(t)} \right| > 0,\,\,t \ge {t_0}$. Define $\tau (t) = \int_{{t_0}}^t {\left| {{\pmb{\omega }}(u)} \right|du} $; this is invertible, so dependent variables may be considered as a function of $t$ or $\tau $. Suppose  $\boldsymbol{E}$ satisfies Equation \ref{eqn: Euler diff eq rep 1} for ${t_0} \le t < {t_1}$.   Let ${\pmb{\tilde \omega }} ={\pmb{ \omega }}/
\left| {{\pmb{\omega }}} \right|.$  Then $\boldsymbol{E}$ satisfies Equation \ref{eqn: Euler diff eq rep 1} with $t$  replaced by $\tau $ and ${\pmb{\omega }}$ replaced by  ${\pmb{\tilde \omega }}$; that is,

\[\frac{{d\boldsymbol{E}}}{{d\tau }} = {\pmb{\tilde \omega }} - \left\{ {\pmb{\tilde \omega }}{|\boldsymbol{E}|^2} - \boldsymbol{E}\left( {{\pmb{\tilde \omega }} \cdot \boldsymbol{E}} \right)\right\}  \left( {\frac{{1 - g\left( {\left| \boldsymbol{E} \right|} \right)}}{{{{\left| \boldsymbol{E} \right|}^2}}}} \right) + \,\,{\pmb{\tilde \omega }} \times \boldsymbol{E}/2 .\]
\end{lemma}
\begin{proof}
This is just the chain rule, using $d\tau/dt = \left| {{\pmb{\omega }}} \right|.$
\end{proof}
Thus in Theorem \ref{thm: spinor}, the hypothesis could be changed to ${\pmb{\omega }}/\left| {\pmb{\omega }} \right|$ being constant, and the conclusions would hold using the new time variable, and the statement about ${\bf{n}}$ changing sign every time the body makes one complete rotation is still true.  In general, the geometrical behavior of $\boldsymbol{E}$ is not affected by a change of time scale, in the same way that the geometry of a curve is not affected by a change of parameterization.  This is useful in our computer experiments with $\boldsymbol{E}$: the magnitude of ${\pmb{\omega }}$ might as well be kept fixed.\\

\section{\textbf{Continuous evolution of axis and angle for all time; continuation of $\bf{E}$ at boundaries.}}
\label{sec: boundary}

The example given in equation (\ref{eqn: pathological example}) showed that for arbitrary continuous ${\pmb{\omega }}$, it is not necessarily possible to define a unit vector $\bf{n}$ for the Euler rotation axis and rotation angle $\theta$, such that both are continuous for all time. For that example, the solution for $\bf{E}$, or for any of the generalized Euler vectors, exists and is continuous at rotation angle $\theta = 0$, but $\bf{n}$ cannot be defined to be continuous there. The same problem can exist at multiples of $2\pi$. One can make a similar example such that a solution ${\bf{E}}(t)$ for equation \ref{eqn: Euler diff eq rep 1} which starts with $0 \le \left| {{\bf{E}}({t_0})} \right| < 2\pi $, has  $\left| {{\bf{E}}(t)} \right|$ approach $2\pi $ in finite time; that is, $\mathop {\lim }\limits_{t \to {{t_1}^-}} \left| {{\bf{E}}(t)} \right| = 2\pi $, for some finite ${t_1}$. But $\mathop {\lim }\limits_{t \to {t_1}^-} {{\bf{n}}(t)}$ does not exist, nor does $\mathop {\lim }\limits_{t \to {t_1}^-} {{\bf{E}}(t)}$: they spin with increasing speed as $t \rightarrow {t_1}^-$, in that example.

However, such pathology does not exist if we put mild regularity conditions on ${\pmb{\omega }}$.
\begin{theorem}[Continuous evolution of Euler axis and angle]\label{thm:continuous axis and angle}
Let ${\pmb{\omega }}$ be the angular velocity function of a rigid body for which the origin of the body system is constrained to be at the origin of the space system. Assume that for all $t \ge t_0$, there exists a non-negative integer $i$ such that ${{\pmb{\omega }}^{(i)}}(t) \ne {\bf{0}}$. Then ${{\bf{n}}(t)}$ and $\theta(t)$ can be defined with ${{\bf{n}}}$ continuous and $\theta$ differentiable for all $t \ge t_0$, such that $Rot_{{\bf{n}}(t),\theta (t)}$ is the position of the rigid body at time $t$.

If in addition it is assumed that for all $t \ge t_0$, there exists a non-negative integer $i$ such that ${{\pmb{\omega }}^{(i)}}(t) \ne {\bf{0}}$ and ${{\pmb{\omega }}^{(i+1)}}(t)$ exists, then ${{\bf{n}}}$ is differentiable for all $t \ge t_0$.
\end{theorem}

Here is a trivial example showing that $\theta$ must be allowed to take arbitrary positive as well as negative values, in order to satisfy the conclusions of the theorem.  Suppose $\theta(t_0) = 0$. Let ${\pmb{\omega }}(t) = s(t){\bf{k}}$, where $s(t_0)$ is positive, then gradually decreases to zero, and becomes negative and stays negative. The body just rotates about the vertical, at first counterclockwise, then slowing down and eventually rotating clockwise. Then ${{\bf{n}}}(t) = {\bf{k}}$, and $\theta(t)$ increases to a positive amount that can be arbitrarily large, then decreases to zero, and then must go negative in order that ${\bf{n}}$ remain continuous. We could keep $\theta$ positive only by discontinuously flipping ${\bf{n}}$ to $-{\bf{k}}$; the Euler vector ${\bf{E}} = {\bf{n}}\theta$ does not distinguish between these.

We'll prove the theorem by analyzing the solution to equation (\ref{eqn: quaternion diffeq}), which exists for all time, at points where ${\bf{M}} = {\bf{0}}$. Under the hypotheses of Theorem \ref{thm:continuous axis and angle}, these points are isolated.  The trick will be to choose the correct branch of the $\arccos$ function in Theorem \ref{thm: main} between these points so that axis and angle connect up continuously.

\begin{lemma}[Zeros are isolated]
\label{lemma: isolated zeros}
Let ${\pmb{\omega }}$ satisfy the hypotheses of Theorem \ref{thm:continuous axis and angle}, and let $({M_0(t),\bf{M}}(t))$ solve equation (\ref{eqn: quaternion diffeq}) for ${t_0} \le t$, with $M_0(t_0)^2 + {\bf{M}}(t_0) \cdot {\bf{M}}(t_0)  = 1$. If $t_0 < t_1$ and ${\bf{M}}(t_1) = {\bf{0}}$, there exists $\delta > 0$ such that ${\bf{M}}(t) \ne {\bf{0}}$ for $0 < |t-t_1| < \delta$.
\end{lemma}
\begin{proof}
Let $i = min\{j \ge 0 :{\pmb{\omega }}^{(j)}(t_1) \ne {\bf{0}}\}$. Then ${\pmb{\omega }}(t)={\pmb{\omega }}^{(i)}(t_1){h^i}/i! + o(h^i)$, where $h=t-t_1$. Now $M_0(t) = M_0(t_1) + \dot{M_0}(t_1)h+ o(h) = M_0(t_1) + o(h)$ since $\dot{M_0}(t_1) = {\pmb{\omega }}(t_1) \cdot {{\bf{M}}}(t_1) = 0$. Equation (\ref{eqn: quaternion diffeq}) immediately implies ${|\bf{M}}(t)| = O(h)$. Then (\ref{eqn: quaternion diffeq}) gives ${\dot{\bf{M}}}(t) = {\pmb{\omega }}^{(i)}(t_1)({h^i}/i!)M_0(t_1)/2 + +o(h^i)$. Integrating, ${{\bf{M}}}(t) = {\pmb{\omega }}^{(i)}(t_1)({h^{i+1}}/(i+1)!)M_0(t_1)/2 + +o(h^{i+1})$.  Since\\ ${\pmb{\omega }}^{(i)}(t_1) \ne {\bf{0}}$ and $|M_0(t_1)| = 1$, the conclusion of the lemma follows.
\end{proof}

The next lemma is the key tool in the proof of the theorem.  It shows precisely how to continue the axis and angle continuously when the angle reaches a multiple of $2\pi$, where ${\bf{M}} = {\bf{0}}$. 
\begin{lemma}[Continuation at a zero] 
\label{lemma: continuous continuation at a zero}
Let ${\pmb{\omega }}$ satisfy the hypotheses of Theorem \ref{thm:continuous axis and angle}, and let \\ $({M_0(t),\bf{M}}(t))$ solve equation (\ref{eqn: quaternion diffeq}) for ${t_0} \le t$, with $M_0(t_0)^2 + {\bf{M}}(t_0) \cdot {\bf{M}}(t_0)  = 1$. Suppose $t_0 < t_1 < t_2$, ${\bf{M}}(t_1) = {\bf{0}}$, and ${\bf{M}}(t) \ne {\bf{0}}$ if $t \ne t_1$ and $t_0 < t < t_2$. Suppose, for some integer $l$, $\theta(t)$ has been defined for $t_0 \le t \le t_1$ by $\cos (\theta(t)/2) = M_0(t)$ with $2\pi l \le \theta(t) \le 2\pi (l+1)$; and for $t_0 < t < t_1$, ${\bf{n(t)}} = {\bf{M}}(t)/\sin(\theta(t)/2)$. Let $i = min\{j \ge 0 :{\pmb{\omega }}^{(j)}(t_1) \ne {\bf{0}}\}$.

Case (a): $i$ is even.  If $\theta(t_1) = 2\pi (l+1)$, define $\theta(t)$ by $\cos(\theta(t)/2)=M_0(t)$ with $2\pi (l+1) \le \theta(t) \le 2\pi (l+2)$, for $t_1 < t \le t_2$. If $\theta(t_1) = 2\pi l$, define $\theta(t)$ by $\cos (\theta(t)/2) = M_0(t)$ with $2\pi (l-1) \le \theta(t) \le 2\pi l$, for $t_1 < t \le t_2$. Thus $\theta(t)$ crosses the boundary into the neighboring interval in this case, going up if it was at the right end, down if it was at the left end of the previous interval.

Case (b): $i$ is odd. For $t_1 < t \le t_2$ define $\theta(t)$ by $\cos (\theta(t)/2) = M_0(t)$ with $2\pi l \le \theta(t) \le 2\pi (l+1)$. Thus in this case, $\theta(t)$ reflects off the boundary, back into the same interval.

For $t_1 < t < t_2$, define ${\bf{n}}(t) = {\bf{M}}(t)/\sin(\theta(t)/2)$. Then\\ $\mathop {\lim }\limits_{t \to {t_1}^-} {{\bf{n}}(t)} = \mathop {\lim }\limits_{t \to {t_1}^+} {{\bf{n}}(t)} = \{{\pmb{\omega }}^{(i)}(t_1)/|{\pmb{\omega }}^{(i)}(t_1)|\}M_0(t_1)(-1)^{l+i+1}$.  Defining ${\bf{n}}(t_1)$ to be this limit makes ${\bf{n}}$ continuous at $t_1$; so $\dot \theta = {\pmb{\omega }} \cdot {\bf{n}}$ is continuous at $t_1$ also. If ${\pmb{\omega }}^{(i+1)}(t_1)$ exists, then $\dot{{\bf{n}}}$ exists and is continuous at $t_1$ as well.
\end{lemma}
\begin{proof}
${{\bf{M}}}(t) = {\pmb{\omega }}^{(i)}(t_1)({h^{i+1}}/(i+1)!)M_0(t_1)/2 + o(h^{i+1})$, from the proof of Lemma \ref{lemma: continuous continuation at a zero}. So \\${{\bf{M}}}(t)/|{{\bf{M}}}(t)| = ({\pmb{\omega }}^{(i)}(t_1)/|{\pmb{\omega }}^{(i)}(t_1)|)M_0(t_1)sgn(h^{i+1}) + o(1)$. Now $\theta(t)$ has already been chosen for $t_0 \le t \le t_1$, and is to be chosen for  $t_1 < t \le t_2$. For for $t_0 < t < t_2$ and $t \ne t_1$, we need ${{\bf{M}}}(t) = {{\bf{n}}}(t)\sin(\theta(t)/2)$, so $|{{\bf{M}}}(t)| = |\sin(\theta(t)/2)|$, and ${{\bf{n}}}(t) = {{\bf{M}}}(t)/\sin(\theta(t)/2) = ({{\bf{M}}}(t)/|{{\bf{M}}}(t)|) sgn(\sin(\theta(t)/2))$$\\= ({\pmb{\omega }}^{(i)}(t_1)/|{\pmb{\omega }}^{(i)}(t_1)|)M_0(t_1)
sgn\{h^{i+1}\sin(\theta(t)/2)\} + o(1) $. In order that ${\bf{n}}(t)$ have the same limit from the right and left as $t \rightarrow t_1$, $\sin(\theta(t)/2)$ must change sign as $t$ crosses $t_1$ if $i$ is even, and not if $i$ is odd. Note $\theta(t_1)$ must be either $2\pi l$ or $2\pi(l+1)$. If $i$ is even, continuity of $\theta$ then requires that, for $t_1 < t \le t_2$, it be chosen with $2\pi (l-1) \le \theta(t) \le 2\pi l$ if $\theta(t_1) = 2\pi l$, and $2\pi (l+1) \le \theta(t) \le 2\pi (l+2)$ if $\theta(t_1) = 2\pi (l+1)$. If $i$ is odd, $\theta$ must not change sign, so continuity requires $2\pi (l+1) \le \theta(t) \le 2\pi (l+2)$ for $t_1 < t \le t_2$.

Now $sgn(\sin(\theta(t)/2) = (-1)^l$ for $t_0 < t < t_1$. If $i$ is even, $sgn(h^{i+1}) = -1$ and \\ $sgn\{h^{i+1}\sin(\theta(t)/2)\} = (-1)^{l+1}$ for $t_0 < t < t_1$. For $t_1 < t < t_2$, $h > 0$ but $\sin(\theta(t)/2) = (-1)^{l+1}$, so again $sgn\{h^{i+1}\sin(\theta(t)/2)\} = (-1)^{l+1}$.  If $i$ is odd, $sgn\{h^{i+1}\sin(\theta(t)/2)\} = (-1)^l$ on both sides of $t_1$.  We can summarize with 
 ${\bf{n}}(t) =({\pmb{\omega }}^{(i)}(t_1)/|{\pmb{\omega }}^{(i)}(t_1)|)M_0(t_1)(-1)^{l+i+1} + o(1)$, whose limit as $t \rightarrow t_1$ is as asserted.

The proof of existence of $\dot{\bf{n}}$ at $t_1$ is similar, but much longer and more tedious. We'll leave out some detail. By the stronger assumption, ${\pmb{\omega }}(t)={\pmb{\omega }}^{(i)}(t_1){h^i}/i! + {\pmb{\omega }}^{(i+1)}(t_1){h^{i+1}}/(i+1)! + o(h^{i+1})$. From earlier, ${{\bf{M}}}(t) = {\pmb{\omega }}^{(i)}(t_1)({h^{i+1}}/(i+1)!)M_0(t_1)/2 + o(h^{i+1})$.   Since ${\pmb{\omega }}^{(i)}(t_1) \times {\pmb{\omega }}^{(i)}(t_1) =0$, we get ${\pmb{\omega }}(t) \times {{\bf{M}}}(t) = o(h^{2i+1}$. Using this and $M_0(t) = M_0(t_1) + o(h)$ in (\ref{eqn: quaternion diffeq}),$\dot{{\bf{M}}}(t) = $ $ {\pmb{\omega }}^{(i)}(t_1)({h^{i}}/i!)M_0(t_1)/2 + {\pmb{\omega }}^{(i+1)}(t_1)({h^{i+1}}/(i+1)!)M_0(t_1)/2 + o(h^{i+1})$. Integrating, \\
${\bf{M}}(t) = {\pmb{\omega }}^{(i)}(t_1)({h^{i+1}}/(i+1)!)M_0(t_1)/2 + {\pmb{\omega }}^{(i+1)}(t_1)({h^{i+2}}/(i+2)!)M_0(t_1)/2 + o(h^{i+2})$. 

Now ${\bf{n}}(t) = ({\bf{M}}(t)/|{\bf{M}}(t)|) sgn(\sin(\theta(t)/2))$ for $t \ne t_1$. Calculate\\ ${\dot{\bf{n}}}(t) = \{{\dot{\bf{M}}}(t)({\bf{M}}(t) \cdot {\bf{M}}(t)) - {\bf{M}}(t)({\bf{M}}(t) \cdot {\dot{\bf{M}}}(t))\}sgn(\sin(\theta(t)/2))/|{\bf{M}}(t)|^3$. Straightforward but very tedious algebra which we omit, using the above for ${\bf{M}}$ and $\dot{{\bf{M}}}$, gives 
\\ ${\dot{\bf{n}}}(t) = 
\left\{{\pmb{\omega }}^{(i+1)}(t_1)|{\pmb{\omega }}^{(i)}(t_1)|^2 -  {\pmb{\omega }}^{(i)}(t_1)\left({\pmb{\omega }}^{(i)}(t_1) \cdot {\pmb{\omega }}^{(i+1)}(t_1)\right)\right\}$\\ 
$\times \left\{(i+1)/\left((i+2)|{\pmb{\omega }}^{(i)}(t_1)|^3 \right) \right\} M_0(t_1)(-1)^{l+i+1} + o(1)$, whose limit as $t \rightarrow t_1$ exists. By elementary real analysis, that implies ${\dot{\bf{n}}}(t_1)$ exists and is equal to that limit.
\end{proof}

\begin{proof}[Proof of Theorem \ref{thm:continuous axis and angle}]
This follows by induction from Lemmas \ref{lemma: isolated zeros} and \ref{lemma: continuous continuation at a zero}. Suppose $t_0 < t_1 <...< t_n$ satisfy ${\bf{M}}(t_j) = {\bf{0}}$ for $j=1,...,n$, and ${\bf{M}}(t) \ne {\bf{0}}$ for $t_0 < t < t_n$ and $t \ne t_j$ for $j=1,...,n$.   Assume ${\bf{n}}(t)$ and $\theta(t)$ have been defined continuously for $t_0 \le t \le t_n$. Let $t_{n+1} = \inf\{t>t_n : {\bf{M}}(t) = {\bf{0}}\}$, where $t_{n+1} = \infty$ if no such $t$ exists.  By Lemma \ref{lemma: isolated zeros}, $t_{n+1} > t_n$.  By Lemma \ref{lemma: continuous continuation at a zero}, ${\bf{n}}(t)$ and $\theta(t)$ can be continuously extended to $t_0 \le t \le t_{n+1}$.  By Lemma \ref{lemma: isolated zeros}, there can be no cluster point of the $t_n$, so "Zeno's paradox" cannot occur; that is, $t_n \rightarrow \infty$, and ${\bf{n}}(t)$ and $\theta(t)$ are defined for all $t$.
\end{proof}

\begin{corollary}[Continuation of ${\bf{E}}$ at boundaries] Let ${\pmb{\omega }}$ satisfy the hypotheses of Theorem \ref{thm:continuous axis and angle}. Then there exists ${\bf{E}}(t)$ which is continuous for all $t \ge t_0$, and satisfies Equation (\ref{eqn: Euler diff eq rep 1}) except where $|{\bf{E}}(t)| = |\theta(t)|$ is a positive multiple of $2\pi$ and the r.h.s of the equation is not defined. If ${\pmb{\omega }}$ satisfies the stronger differentiability condition of the theorem, ${\bf{E}}(t)$ is differentiable for all $t \ge t_0$
\end{corollary}
\begin{proof}
Let ${\bf{E}}(t) = {\bf{n}}(t) \theta(t)$, where ${\bf{n}}$ and $\theta$ are as constructed in the proof of the theorem. This ${\bf{E}}$ is continuous everywhere, and also differentiable everywhere when ${\bf{n}}$ is, since $\dot{\bf{E}} = \dot{\bf{n}}\theta + {\bf{n}}\dot{\theta}$.
\end{proof}
Thus it represents a continuation of solutions of the Euler vector equation if a boundary is encountered, under these hypotheses.  We could have constructed this continuation arguing directly from that equation, but the argument using \ref{eqn: quaternion diffeq} was slightly simpler. 

\begin{exa}
[Non-trivially going through $2\pi$]
To construct this, we use the trick of running the equation for ${\bf{E}}$ backwards, in a certain way.  Let ${\pmb{\omega }}(t) = {\bf{i}} + t{\bf{j}}$ and ${\bf{E}}(0) = {\bf{0}}$. Run Equation (\ref{eqn: Euler diff eq rep 1}) from ${t=0}$ to $t_1 =1$, and let $\theta_1 = |{\bf{E}}(t_1)|$ and ${\bf{n}}_1 = {\bf{E}}(t_1)/|{\bf{E}}(t_1)|$.

We'll use the endpoint of this first run to create initial conditions for a second run.  Let ${\pmb{\omega^*}}(t) = -{\pmb{\omega}}(t_1 - t) = {-\bf{i}} - (t_1 - t){\bf{j}}, t\ge 0$.  Let $\theta^{*}(0) = 2\pi - \theta_1$, ${\bf{n}}^{*}(0) = -{\bf{n}}_1$, and ${\bf{E}}^{*}(0) = {\bf{n}}^{*}(0)\theta^{*}(0)$. Now run Equation (\ref{eqn: Euler diff eq rep 1}) for ${t \ge 0}$, using ${\pmb{\omega^*}}$ and ${\bf{E}}^{*}$ with these initial conditions.  It can be shown that then  $\theta^{*}(t) = 2\pi - \theta(t_1 - t)$ and ${\bf{n}}^{*}(t) = -{\bf{n}}(t_1 - t)$ for $0 \le t \le t_1$, where $\theta^{*}(t) = |{\bf{E}}^{*}(t)|$ and ${\bf{n}}^{*}(t) = {\bf{E}}^{*}(t)/|{\bf{E}}^{*}(t)|$. Now $\theta(t_1 - t)$ approaches zero as $t \rightarrow t_1^-$, so $\theta^{*}(t)$ approaches $2\pi$, and we have our example.   Since ${\pmb{\omega^*}}(t_1) \ne {\bf{0}}$, $\theta^{*}(t)$ should go past $2\pi$ as $t$ passes $t_1$, according to Theorem \ref{thm:continuous axis and angle}.
\end{exa}

\begin{figure}
\centering
	\includegraphics[width=0.8\textwidth,height=0.6\linewidth]{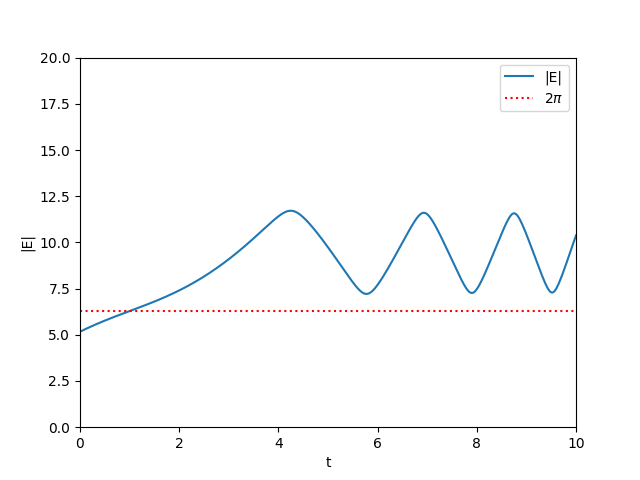}

	\caption{Numerical solution for the example in which $|{\bf{E}}|$ is made to pass through 2$\pi$. It then continues to oscillate indefinitely afterwards.}
	\label{fig:going_through_2pi}
\end{figure}

Figure \ref{fig:going_through_2pi} shows the result of solving Equation (\ref{eqn: Euler diff eq rep 1}) numerically for this example.  It goes through $2\pi$ with no difficulty, even though the r.h.s. of the equation is not actually defined at exactly $t_1$, as long as $t_1$ exactly is not a node. Near $t_1$, the middle term of the equation has a factor going to infinity, but the other factor goes to zero so that the term remains bounded, and it is well-behaved numerically. After passing through the 2$\pi$ boundary, the solution simply continues to oscillate in-between the 2$\pi$ and 4$\pi$ regions, as shown.

\section{\textbf{Numerical solutions of the Differential Equation.}}
\label{sec:NumericalSolns}
The dynamical system for the evolution of the Euler vector, Equation 
 \ref{eqn: Euler diff eq rep 1}, or one of the alternate representations such as the Modified Gibbs \ref{eqn: quaternion diffeq}, is generally non-autonomous, specified through the choice of the angular velocity function $\pmb{\omega}(t)$. In the simplest possible scenario in which $\pmb{\omega}$ is constant in time, making the system autonomous, it was shown in Section \ref{sec: spinor} that there exist periodic solutions for $\boldsymbol{E}$ when $\boldsymbol{E}$ is not initially parallel to $\pmb{\omega}$. A visualization of such a periodic trajectory was shown in Figures \ref{fig:Euler vec const w trajectory} and \ref{fig:Euler vec const w ind variables} and briefly described in the Introduction.

For the more general case when $\pmb{\omega}(t)$ is not constant in time, we resort to numerical integration to study the behavior of the Euler vector. In particular, we analyze a case in which the angular velocity vector itself rotates in a plane at constant frequency. This prototype example may in fact capture some quite general properties of how $\boldsymbol{E}$ may behave. The evolution of the Euler vector in this case produces some seemingly complex and intricate trajectories, though by applying techniques from dynamical systems theory, we will be able to conclude that the motion is in fact quasiperiodic. 

\subsection{Dynamical system properties.}
We first may be interested in what kinds of general, global statements we can make about the dynamical system.
With the Gibbs representation, Equation \ref{eqn: Gibbs vec diffeq}, it becomes especially simple to compute the \textit{divergence} of the vector field describing this dynamical system. Using the notation $\dot{\boldsymbol{G}} = \boldsymbol{f}(\boldsymbol{G})$, it can be easily checked that the divergence of $\boldsymbol{f}$ is found to be

\begin{equation}
\label{divergence}
\nabla \cdot \boldsymbol{f} = \frac{\partial f_{x}}{\partial G_x} + \frac{\partial f_{y}}{\partial G_y} + \frac{\partial f_{z}}{\partial G_z} = 2\ \omega_x G_x + 2\ \omega_y G_y + 2\ \omega_z  G_z = 2\ \pmb{\omega} \cdot \boldsymbol{G}.
\end{equation}

And since $\dot{\theta} = \pmb{\omega} \cdot \bf{n}$, we could also write 

\begin{equation}
\label{divergence result 2}
\nabla \cdot \boldsymbol{f} = \pmb{\omega} \cdot 2\tan(\theta/2)\textbf{n} = 2\tan(\theta/2) \dot{\theta}.
\end{equation}

At least one reason for being interested in the divergence is that it provides information regarding these global properties of the system \cite{strog} \cite{guck_holmes} \cite{guido}; for a conservative vector field the divergence would vanish everywhere and phase space volumes for the flow are invariant, and for a dissipative system the divergence would be strictly negative, implying a shrinking phase space volume and the possibility for chaotic attractors. This 3D system is generally nonautonomous, however following a common trick and introducing a fourth variable to the system (such as $G_{\tau} = t$, $\dot{G_{\tau}} = 1$, $\partial f_{\tau} / \partial G_{\tau} = 0$) to remove the explicit time dependence, we still get the same result for the divergence.

Just from the first example shown in Figure \ref{fig:Euler vec const w trajectory}, it is clear that $\pmb{\omega} \cdot \boldsymbol{E}$ is neither always zero, positive, or negative, even in the autonomous case of constant $\pmb{\omega}$. Since this would also be true for the Gibbs representation, we are thus not able to classify this dynamical system as being conservative or dissipative, and furthermore this means we are unable to use theorems or results that apply to dynamical systems with one of these special properties.

\subsection{Rotating angular velocity vector.}
We now turn to analyzing a special case of the differential equation (\ref{eqn: Euler diff eq rep 1}) in which 
 we let $\pmb{\omega}(t)$ rotate in a plane at fixed frequency. The subsequent motion of $\boldsymbol{E}$, as shown in Figures \ref{fig:Euler vec rotating w} and \ref{fig:E and strobe period pi}, appears to be considerably more complex. Let us take the parameterization $\pmb{\omega}(t)$ to be given by
\begin{equation}
\label{eqn:omega_const_rotation}
(\omega_x, \omega_y, \omega_z) = (\text{cos}(\alpha t), \text{sin}(\alpha t), 0)
\end{equation}
so that $\pmb{\omega}(t)$ has fixed unit length and rotates in the x-y plane with angular frequency $\alpha$, and period $T_{\alpha} = 2\pi/\alpha$.

Once this functional form for the angular velocity vector is chosen, the only free choices left for investigating solutions of Equation \ref{eqn: Euler diff eq rep 1} are $\alpha$ and the initial condition on $\boldsymbol{E}$. Using the original representation, the length of the vector $\boldsymbol{E}$ has a meaningful interpretation as the total rotation angle a body would move through while rotating about $\boldsymbol{E}$, so we note that any non-zero initial condition, $\boldsymbol{E}_0$, may be thought of as an initial rotation of the body that is needed to get into the position it has when time starts. Throughout this section we use $\boldsymbol{E}_0 = (1/\sqrt{3}, 1/\sqrt{3}, 1/\sqrt{3})$ as the initial condition, though our experiments of repeating the analysis with different nonzero initial conditions show that the qualitative results do not depend strongly on this choice.

Integrating the differential equation with a 4th order Runge-Kutta implementation, we can trace out trajectories of $\boldsymbol{E}(t)$ in time. In Figure \ref{fig: E torus 40} we show the integrated trajectory for $\boldsymbol{E}$ up to $t = $ 4200 time units, as referenced in the Introduction. The value of $T_{\alpha}$ used in this example is 40. In Figure \ref{fig:E_t rotating w} we show length of $\boldsymbol{E}(t)$ against $t$ as a time series for two different times, remembering that in the original representation, the length of $\boldsymbol{E}(t)$ represents the rotation angle of the body about the rotation axis. From the trajectory and time series plots alone, it is certainly not clear whether the motion is periodic, chaotic, or quasiperiodic.

 Recall that when $\pmb{\omega}$ was fixed, there was a natural period with which $\boldsymbol{\hat{E}}$ made periodic revolutions, namely 4$\pi$. Indeed with a rotating $\pmb{\omega}(t)$,  multiple competing frequencies are involved, and below we will look at the spectral components of the motion. 
We know that the body vector rotates around $\pmb{\omega}$ as $\pmb{\omega}$ rotates in a plane, so  $\boldsymbol{B}$ should trace out a torus. But what about the Euler vector? In Figures \ref{fig: E torus 40} and \ref{fig: E torus norm} and the associated animation, we certainly observe a winding torus-like shape, but it is not immediately clear what frequencies are involved to produce this shape. One way to think about $\alpha$ is as a type of "driving" frequency which will produce multiple response frequencies in the behavior of $\boldsymbol{E}$.  Using different values for $T_{\alpha}$ may slightly change the shape or level of deformation of this object, though we find that the general conclusions for the motion are not changed by varying this period, nor by changing the starting condition $\boldsymbol{E}_0$.  Note that these trajectory plots of $\boldsymbol{E}$ are not true "phase space" plots since the system is nonautonomous.

\begin{figure}
\centering
	\begin{subfigure}{0.49\textwidth}
		\includegraphics[width=1.0\textwidth,height=1.0\linewidth]{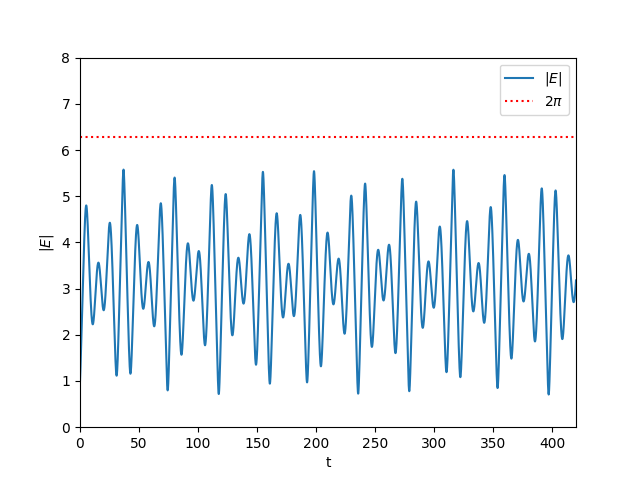}
		\caption{Up to time 420}
        \label{fig: ts_1}
	\end{subfigure}
	\begin{subfigure}{0.49\textwidth}
		\includegraphics[width=1.0\textwidth,height=1.0\linewidth]{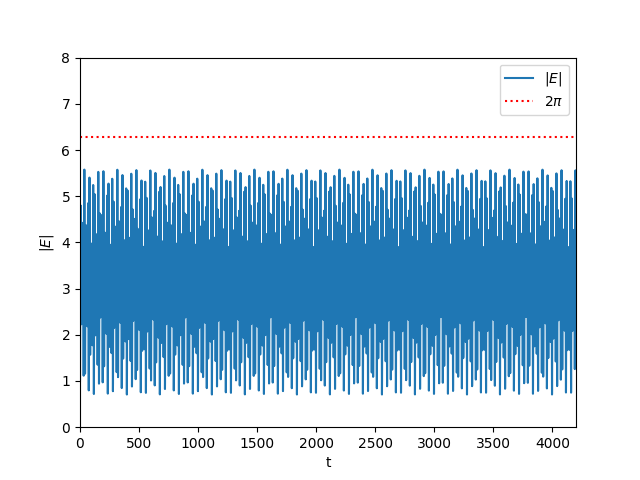}
		\caption{Up to time 4200}
		\label{fig: ts_2}
	\end{subfigure}
	
	\caption{Time series plot of $|\boldsymbol{E}|$ versus $t$ for different lengths of time, for rotating $\pmb{\omega}(t)$, using driving period $T_\alpha$ = 40.}
	\label{fig:E_t rotating w}
\end{figure}

\subsection{Quasiperiodic solutions.}

To further analyze the behavior of trajectories of the Euler vector when $\pmb{\omega}(t)$ rotates at constant frequency, we apply some techniques commonly used in dynamical systems theory; namely analyzing Poincar\'e section and strobe plots, computing Lyapunov exponents, analyzing the spectral frequencies of the motion, and making recurrence plots. These various techniques all lead us to the conclusion that the trajectory of $\boldsymbol{E}$ for this example is \textit{quasiperiodic}. See \cite{das} and \cite{zou_thesis} for a general discussion.

Computing the Lyapunov exponents for a dynamical system helps guide the process by showing the rate of divergence of trajectories starting from nearby initial conditions. Essentially a positive maximum Lyapunov exponent is a strong indicator of chaotic motion, demonstrating sensitivity on the precise initial conditions \cite{strog}. If the maximum Lyapunov exponent is numerically close to 0, indicating that nearby trajectories do not separate at an exponential rate, one may conclude the motion is not chaotic. Evolving a ball of different initial conditions, the dynamical system will stretch and distort the ball, and the asymptotic rates of expansion/contraction in each dimension are the Lyapunov exponents \cite{sandri}  \cite{janaki}. Informally, if the initial separation between trajectories has size $|{\pmb{\delta}}_0|$, then after time $t$ the separation will have grown to $| {\pmb{\delta}}(t)| \approx | {\pmb{\delta}}_0 | e^{\lambda t}$. Computing the ratio of the error terms for long times gives an approximation to $\lambda$, the maximum Lyapunov exponent.
For formal numerical computation of the Lyapunov exponents, one computes the linearized system and Jacobian about an initial point. Following the method described in \cite{kuptsov}, a package provided \cite{savary} is used for the computation. For the Euler vector system, we find that the size of the maximum Lyapunov exponent converges to 0, as seen in Figure \ref{fig:Lyap}, for the case of the constantly rotating $\pmb{\omega}(t)$, indicating that trajectories like the ones shown in Figures \ref{fig: E torus 40}, \ref{fig:E_420_pi}, and \ref{fig:E_4200_pi} are \textit{not} chaotic. After $10^5$ integration steps, the value of the maximum Lyapunov exponent is about $10^{-5}$.

\begin{figure}
\centering
\includegraphics[width=0.85\linewidth]{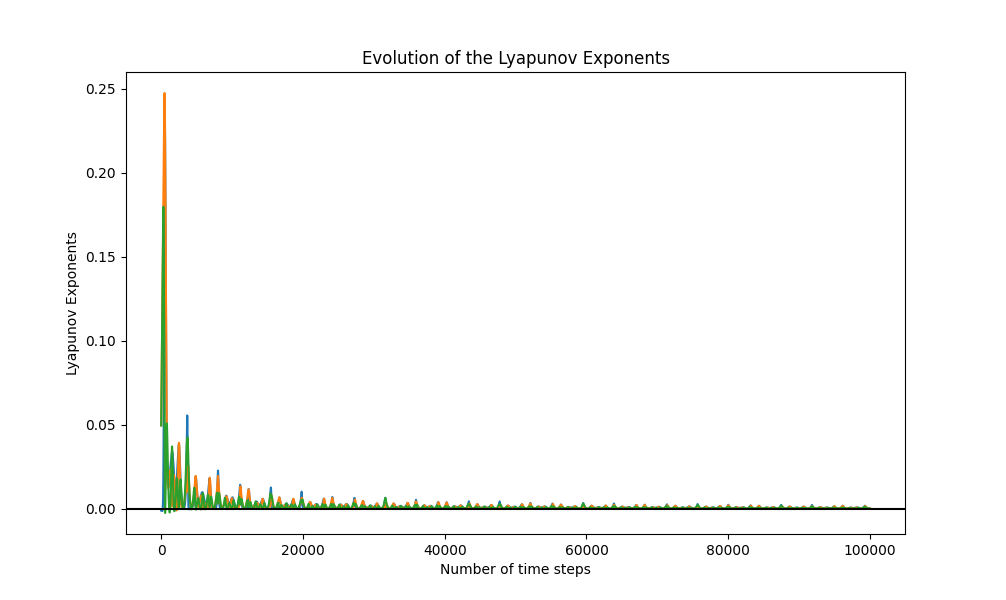}
\caption{\label{fig:Lyap} Convergence of the Lyapunov exponents to 0 for the Euler vector dynamical system in the case of rotating $\pmb{\omega}(t)$, indicating non-chaotic behavior. After $10^5$ steps, the value of the maximum Lyapunov exponent is about $10^{-5}$.}
\end{figure}

\begin{figure}[h!tp]
\centering
	\begin{subfigure}{0.49\textwidth}
		\includegraphics[width=1.0\textwidth,height=1.0\linewidth]{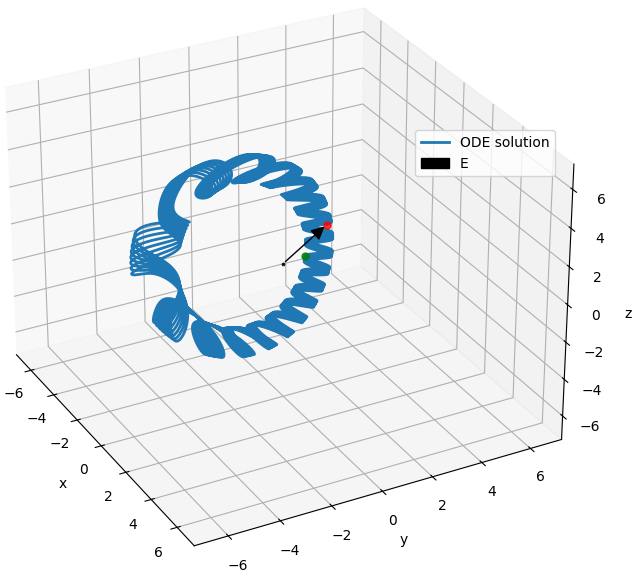}
		\caption{Trajectory of $\boldsymbol{E}$ up to time 420 for $T_\alpha = \pi$}
        \label{fig:E_420_pi}
	\end{subfigure}
	\begin{subfigure}{0.49\textwidth}
		\includegraphics[width=1.0\textwidth,height=1.0\linewidth]{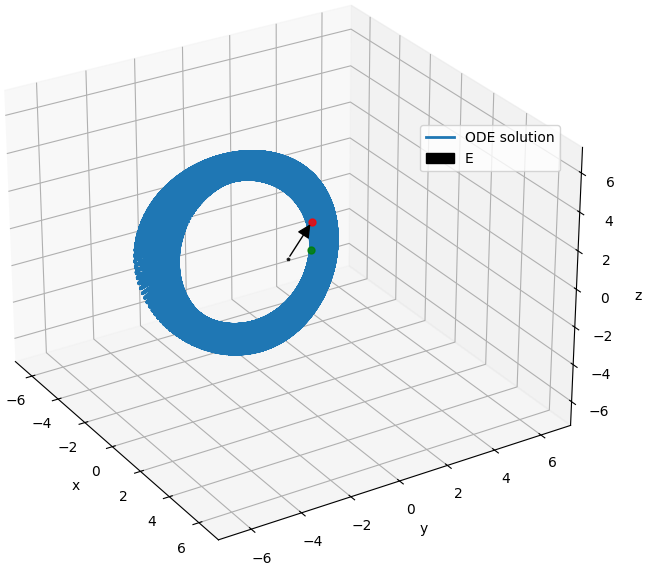}
		\caption{Trajectory of $\boldsymbol{E}$ up to time 4200 for $T_\alpha = \pi$}
		\label{fig:E_4200_pi}
	\end{subfigure}
    \begin{subfigure}{0.49\textwidth}
		\includegraphics[width=1.0\textwidth,height=1.0\linewidth]{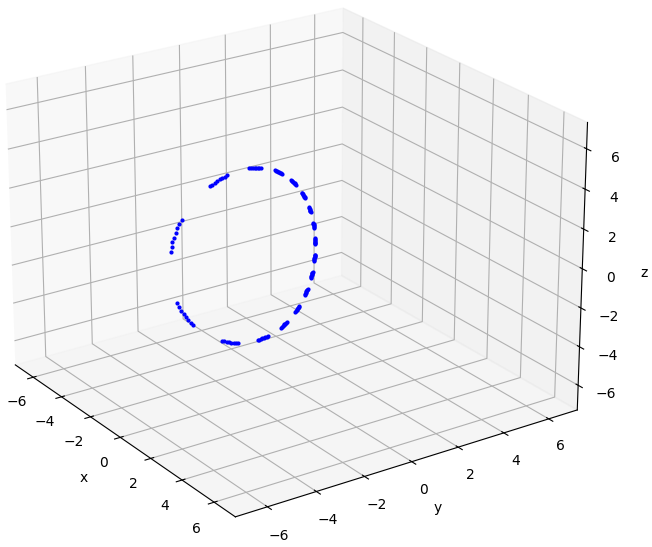}
		\caption{Strobe plot of the trajectory up to time 420}
        \label{strobe_420_pi}
	\end{subfigure}
	\begin{subfigure}{0.49\textwidth}
		\includegraphics[width=1.0\textwidth,height=1.0\linewidth]{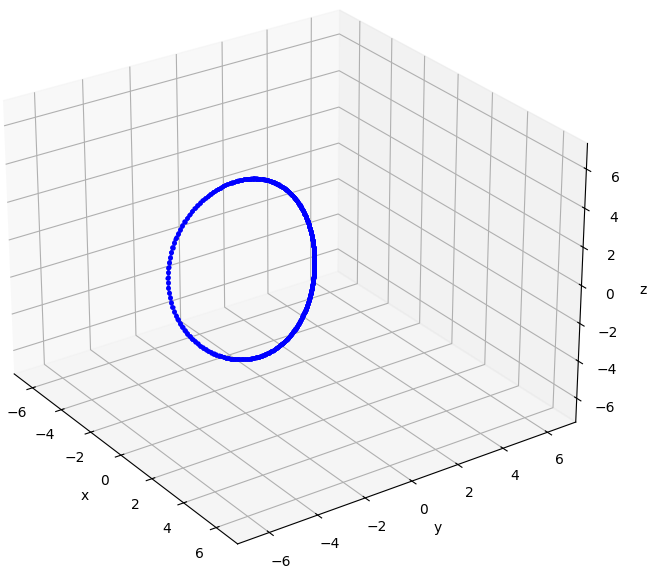}
		\caption{Strobe plot of the trajectory up to time 4200}
		\label{strobe_4200_pi}
	\end{subfigure}
    \begin{subfigure}{0.49\textwidth}
		\includegraphics[width=1.0\textwidth,height=1.0\linewidth]{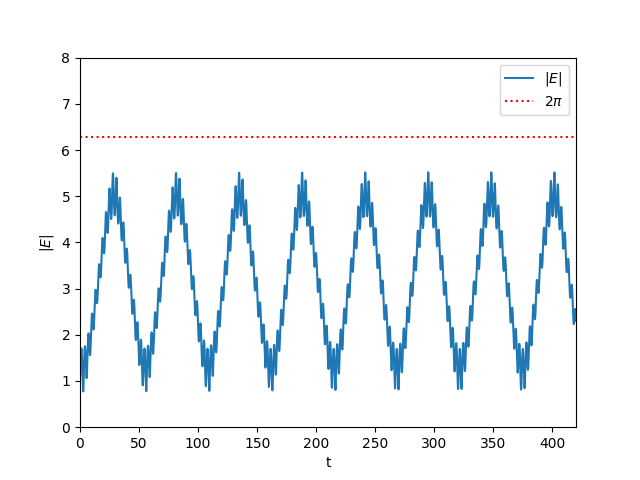}
		\caption{Time series plot of $| \boldsymbol{E} |$ }
	\end{subfigure}
	\begin{subfigure}{0.49\textwidth}
		\includegraphics[width=1.0\textwidth,height=1.0\linewidth]{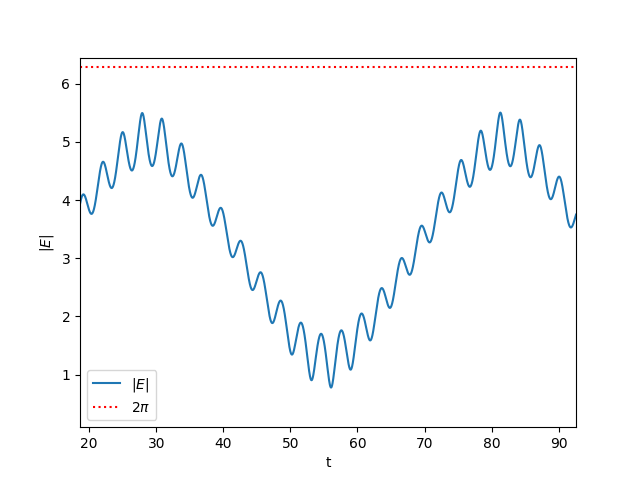}
		\caption{Time series plot, zoomed in.}
		\label{fig:E_t_zoom}
	\end{subfigure}
	\caption{Using period $\pi$ for the rotation of $\pmb{\omega} (t)$ still leads to a quasiperiodic trajectory of $\boldsymbol{E}$. In the zoomed in time series plot, one can approximate the periods of the dominant fast and slow oscillations.}
	\label{fig:E and strobe period pi}
\end{figure}

\begin{figure}
\centering
	\begin{subfigure}{0.49\textwidth}
		\includegraphics[width=1.0\textwidth,height=1.0\linewidth]{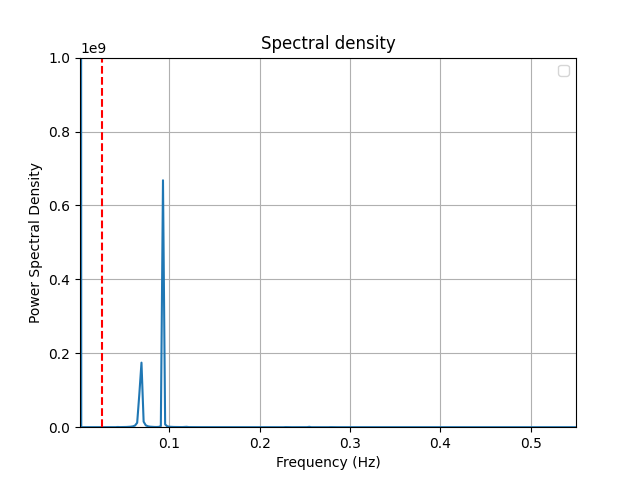}
		\caption{Power spectral density when $T_\alpha = 40$}
        \label{spectral 40}
	\end{subfigure}
	\begin{subfigure}{0.49\textwidth}
		\includegraphics[width=1.0\textwidth,height=1.0\linewidth]{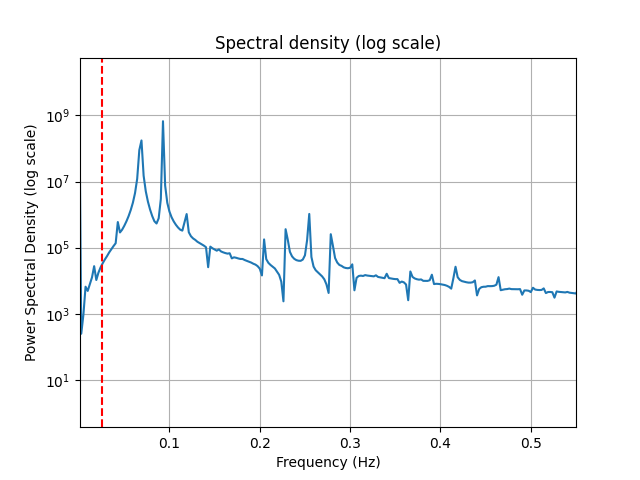}
		\caption{Log scale, power spectral density, $T_\alpha = 40$}
        \label{spectral 40 log}
	\end{subfigure}
	\begin{subfigure}{0.49\textwidth}
		\includegraphics[width=1.0\textwidth,height=1.0\linewidth]{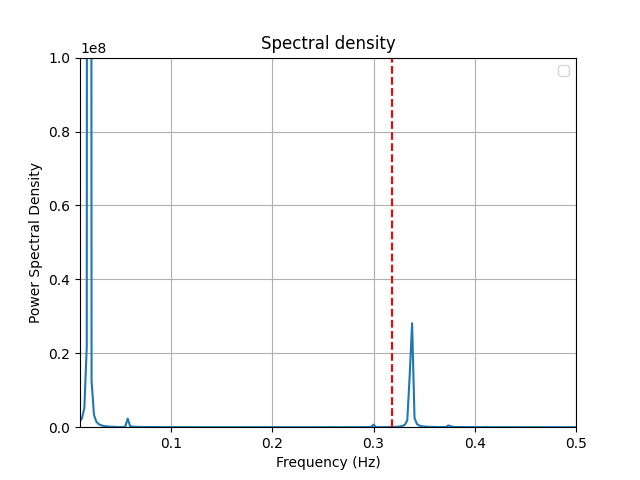}
		\caption{Power spectral density when $T_\alpha = \pi$}
        \label{spectral pi}
	\end{subfigure}
	\begin{subfigure}{0.49\textwidth}
		\includegraphics[width=1.0\textwidth,height=1.0\linewidth]{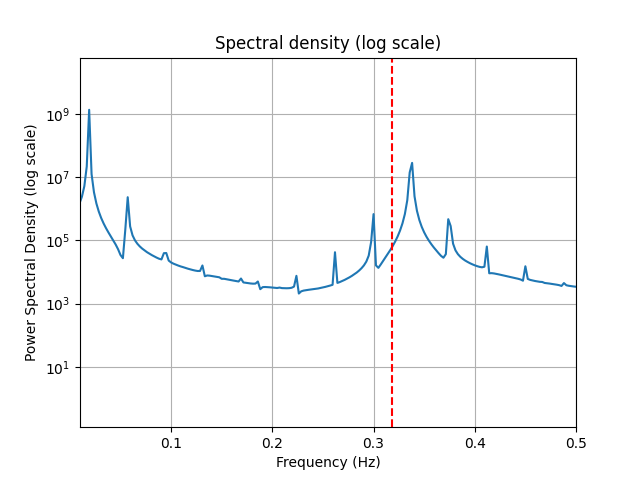}
		\caption{Log scale, power spectral density, $T_\alpha = \pi$}
        \label{ spectral pi log}
	\end{subfigure}
	\caption{Power spectral density plots when $T_\alpha = 40$ and $T_\alpha = \pi$. Dashed vertical lines represent the driving frequency for $\pmb{\omega}(t)$. The dominant response frequencies for $\boldsymbol{E}$ can be seen in the torus-like shape of the full trajectory. }
	\label{fig: spectral 1}
\end{figure}

\begin{figure}[h!]
\centering
\includegraphics[width=0.7\linewidth]{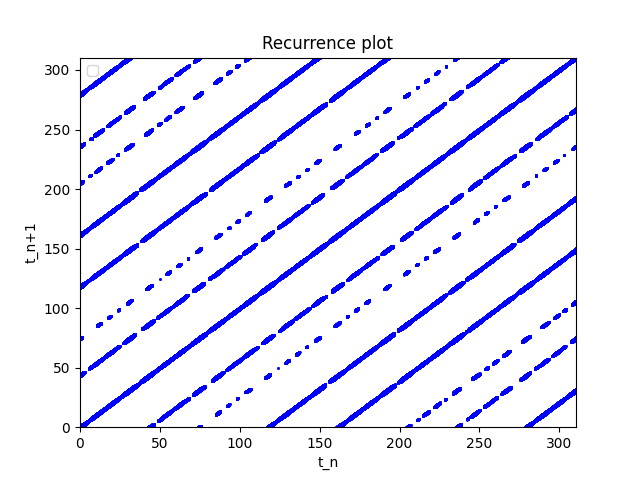}
\caption{\label{fig:recurrence1} Recurrence plot for Euler vector solution for rotating $\pmb{\omega}(t)$ with $T_\alpha = 40$.}
\end{figure}

We next look at a Poincar\'e section plot of the 3D trajectory or flow for $\boldsymbol{E}$, a mapping of the full trajectory onto a lower dimensional space. Taking a slice of the full spatial trajectory along some axis or plane is one common method used for Poincar\'e sections, though in this case a \textit{time strobe} section plot is more relevant. \footnote{For a well-known example, see the discussions in Guckenheimer/Holmes \cite{guck_holmes} and Strogatz \cite{strog} of Duffing's equation for a double well potential as a model for a nonlinear buckled beam, studied by Moon/Holmes \cite{moon_holmes}. In this case the image of the strobe plot (Guckenheimer/Holmes page 90) reveals a fractal set of a strange attractor that has become a notable image for chaotic dynamics.} We plot values of $\boldsymbol{E}(t)$ whenever $t$ is a multiple of the period associated with the driving frequency for $\pmb{\omega}(t)$, $T_{\alpha}$. In Figures \ref{fig: E strobe 420 40} and \ref{fig: E strobe 4200 40} we see such a strobe plot for the example with $T_{\alpha}=40$ after a moderate and a longer time. In the early stages of the strobe plot we see a ring of dots, but as time is increased, it becomes clear that the dots are densely filling in a \textit{closed orbit} in this strobed phase space. Were the flow periodic, there would instead only be a finite number of points on the Poincar\'e map, and if the flow were chaotic, the points would not lie on a simple closed curve but would rather be expected to have some more complicated fractal-like structure. The existence of a smooth, closed curve in the strobe plots indicates quasiperiodicity. The interpretation is that each time $\pmb{\omega}(t)$ returns back to it's initial position, the location of $\boldsymbol{E}$ lies somewhere on this curve, and in fact it's location appears to "jump" around on successive passes, filling in the curve in a rather clumpy way, but never returning to exactly the same point! The full orbit of $\boldsymbol{E}$, then, would be filling out a torus-like shape, as we see in Figure \ref{fig: E torus 40}.

The driving frequency associated with the rotation rate of $\pmb{\omega}$ certainly plays a role in the subsequent motion of $\boldsymbol{E}$, but what can we say about the response frequencies involved in producing the toroidal shape in Figure \ref{fig: E torus 40}? For the example we have been using with $T_{\alpha} = 40$, we show a spectral density plot for $\boldsymbol{E}$ (\cite{dumont} \cite{valsakumar}) in Figures \ref{spectral 40} and \ref{spectral 40 log}. The vertical dotted line corresponds to the frequency associated with the rotation of $\pmb{\omega}$, that is, 1/40. The two dominant response frequencies which show up occur at frequencies of about 0.068 and 0.0925, corresponding to periods of about 14.7 and 10.8 for the $T_{\alpha} = 40$ case. These dominant response periods have no obvious connection to any of the input values to the system. Is there perhaps a way to pick the driving frequency or starting condition such that the dominant response frequencies are not \textit{incommensurate}? Our experiments with this have all led to incommensurate response frequencies for $\boldsymbol{E}$, and thus quasiperiodic motion.
As a second example, we pick $T_{\alpha} = \pi$ as the period associated with the driving frequency, with the thought that since $4 \pi$ was the periodicity of $\boldsymbol{\hat{E}}$ when $\pmb{\omega}$ is constant, this might produce a commensurate response and a periodic trajectory. In Figure \ref{fig:E and strobe period pi} we see that this is not the case, and we still end up with incommensurate response frequencies and quasiperiodic motion of $\boldsymbol{E}$. In fact, for this case, from the spectral density plot in Figure \ref{spectral pi}  we see that the dominant frequencies occur at values corresponding to periods of roughly 3 and 53. Looking at the time series plot for $| \boldsymbol{E} |$ in Figure \ref{fig:E_t_zoom}, we can approximately read off the period for the fast and slow oscillations and verify that the values match what was found from the spectral density plot.

As a final method of analysis from the dynamical systems toolbox, we look at the behavior of the numerically integrated solution of the Euler vector differential equation using a \textit{recurrence plot} \cite{zou_thesis} \cite{zou} \cite{ivchenko}. A recurrence plot shows the times for which a trajectory visits a small neighborhood of each point on the trajectory. In other words, whenever a trajectory returns to within some small distance threshold $\epsilon$ of a point, those points are plotted vertically above that point. This is done for every point along the trajectory, and usually produces bands of diagonal lines for simple systems. One can compute the binary matrix $\boldsymbol{R}_{i,j} = \Theta(\epsilon - ||\boldsymbol{x}_i - \boldsymbol{x}_j||)$ for each pair of points $i$, $j$ on the trajectory where $\epsilon$ is a threshold and $\Theta$ is the Heaviside function. So each matrix entry that is 1 is plotted as a point on the recurrence plot. The takeaway is that for periodic motion, the diagonal lines on the recurrence plot will all be equally spaced. For quasiperiodic motion, one finds lines of different, unequal spacings, reflecting the different almost-periodic time scales, and for chaotic motion, the recurrence plot pattern will break from the band structure into a more complicated shape. (Examples in Chapters 2-4 of \cite{zou_thesis}, in particular Fig 4.6, show a typical distinction of quasiperiodic and chaotic recurrence plots in a real system, the H\'enon-Heiles system). In Figure \ref{fig:recurrence1} we show the recurrence plot for the Euler vector trajectory with $T_\alpha = 40$, using a threshold of $\epsilon = 0.1$. The unequally spaced bands further confirm the quasiperiodic nature of the motion. \\

\bibliographystyle{vancouver}
\bibliography{EulerRotationArxiv}

\end{document}